\documentclass[journal]{IEEEtran}

\usepackage{amsmath,amssymb,amsfonts,bm,mathtools}
\usepackage{algorithm}
\usepackage{algorithmicx}
\usepackage{algpseudocode}
\usepackage{array}
\usepackage{booktabs}
\usepackage{cite}
\usepackage{graphicx}
\usepackage[caption=false,font=footnotesize]{subfig}
\usepackage{xcolor}
\usepackage{url}
\usepackage[colorinlistoftodos]{todonotes}

\newtheorem{assumption}{Assumption}
\newtheorem{theorem}{Theorem}

\newtheorem{lemma}{Lemma}
\newtheorem{remark}{Remark}
\makeatletter
\@ifundefined{proof}{%
}{}
\makeatother

\newcommand{\BA}{{\bold{A}}}
\newcommand{\BB}{{\bold{B}}}

\newcommand{\BD}{{\bold{D}}}
\newcommand{\BE}{{\bold{E}}}
\newcommand{\BF}{{\bold{F}}}
\newcommand{\BG}{{\bold{G}}}
\newcommand{\BH}{{\bold{H}}}
\newcommand{\BI}{{\bold{I}}}

\newcommand{\BP}{{\bold{P}}}
\newcommand{\BQ}{{\bold{Q}}}
\newcommand{\BR}{{\bold{R}}}
\newcommand{\BS}{{\bold{S}}}
\newcommand{\BT}{{\bold{T}}}
\newcommand{\BU}{{\bold{U}}}
\newcommand{\BV}{{\bold{V}}}

\newcommand{\BX}{{\bold{X}}}
\newcommand{\BY}{{\bold{Y}}}
\newcommand{\BZ}{{\bold{Z}}}

\newcommand{\Ba}{{\bold{a}}}
\newcommand{\Bb}{{\bold{b}}}

\newcommand{\Bg}{{\bold{g}}}

\newcommand{\Bn}{{\bold{n}}}
\newcommand{\Bo}{{\bold{o}}}

\newcommand{\Bq}{{\bold{q}}}

\newcommand{\Bs}{{\bold{s}}}

\newcommand{\Bu}{{\bold{u}}}
\newcommand{\Bv}{{\bold{v}}}

\newcommand{\Bx}{{\bold{x}}}
\newcommand{\By}{{\bold{y}}}
\newcommand{\Bz}{{\bold{z}}}

\newcommand{\BGamma}{\bold{\Gamma}}
\newcommand{\BOmega}{\bold{\Omega}}
\newcommand{\BTheta}{\bold{\Theta}}
\newcommand{\BSigma}{\bold{\Sigma}}

\newcommand{\BPi}{\bold{\Pi}}

\newcommand{\BPsi}{\boldsymbol{\Psi}}
\newcommand{\BLambda}{\boldsymbol{\Lambda}}

\newcommand{\Bmu}{\boldsymbol{\mu}}
\newcommand{\Btau}{\boldsymbol{\tau}}

\newcommand{\Beta}{\boldsymbol{\eta}}
\newcommand{\Bdelta}{\boldsymbol{\delta}}
\newcommand{\E}{\mathbb{E}}

\newcommand{\R}{\mathbb{R}}
\newcommand{\C}{\mathbb{C}}

\newcommand{\cC}{\mathcal{C}}

\newcommand{\cN}{\mathcal{N}}
\newcommand{\cO}{\mathcal{O}}
\newcommand{\cP}{\mathcal{P}}

\newcommand{\cR}{\mathcal{R}}

\newcommand{\cV}{\mathcal{V}}

\def\vzero{{\bm{0}}}

\newcommand{\bcC}{\bm{\mathcal{C}}}

\newcommand{\bcR}{\bm{\mathcal{R}}}
\newcommand{\bcT}{\bm{\mathcal{T}}}

\newcommand{\CN}{\mathcal{CN}}

\newcommand{\Prb}{\mathbb{P}}
\newcommand{\Tr}{\operatorname{Tr}}

\newcommand{\blkdiag}{\operatorname{blkdiag}}

\newcommand{\norm}[1]{\left\lVert#1\right\rVert}
\newcommand{\abs}[1]{\left|#1\right|}
\newcommand{\lnorm}[1]{\lVert#1\rVert}

\newcommand{\QED}{\hfill \ensuremath{\blacksquare}}

\newcommand\numberthis{\addtocounter{equation}{1}\tag{\theequation}}

\DeclareMathOperator{\Cov}{Cov}
\DeclareMathOperator{\Var}{Var}

\usepackage{geometry}
\allowdisplaybreaks[4]

\begin{document}

\title{Task-Oriented Precoding for Edge Inference over Large-Scale MIMO Systems}

\author{\IEEEauthorblockN{Hongru~Li, Zeyan~Zhuang, Zixin~Wang, Hengtao~He, Shenghui~Song,\\ Jun Zhang,~\textit{Fellow,~IEEE,} and {Khaled B.~Letaief,~\textit{Fellow,~IEEE}}}
\thanks{Hongru~Li, Zeyan~Zhuang, Zixin~Wang, Shenghui~Song, Jun~Zhang, and Khaled~B.~Letaief are with the Department of Electronic and Computer Engineering, The Hong Kong University of Science and Technology (HKUST), Hong Kong (e-mail: hlidm@connect.ust.hk, zzhuangac@connect.ust.hk, eewangzx@ust.hk, eeshsong@ust.hk, eejzhang@ust.hk, eekhaled@ust.hk).}
\thanks{Hengtao He is with the National Mobile Communications Research Laboratory, School of Information Science and Engineering, Southeast University, Nanjing, 210093, China (e-mail: hehengtao@seu.edu.cn).
}
\vspace{-8mm}
}

\maketitle

\begin{abstract}
Future wireless networks are expected to support networked artificial intelligence (AI) services, where multiple devices transmit learned features to an edge server for distributed inference. This setting calls for task-oriented physical-layer optimization, where wireless transmission should preserve useful information for inference rather than only maximize the rate or reconstruct the transmitted signals. A key physical-layer control variable is the multiple-input multiple-output (MIMO) precoder, which determines how device features are shaped and combined over wireless channels. Existing task-oriented precoding methods typically adapt the precoder to instantaneous channel state information at the transmitter (CSIT). However, in multi-device MIMO systems, acquiring the aggregate channel, feeding back CSI or optimized precoders, and reoptimizing across coherence blocks introduce substantial overhead. This paper develops a random-matrix-theoretic framework based on statistical CSIT that designs a slow-timescale precoder from channel covariance statistics and training-set feature statistics, without requiring instantaneous CSIT. We adopt maximal coding rate reduction (MCR\textsuperscript{2}) to measure the class separability of the received features, yielding a task-aware utility for MIMO precoder design. Since this utility still depends on random small-scale fading, we derive a deterministic approximation that converts it into a fixed-point objective depending only on long-term statistics and large-system dimension ratios by leveraging random matrix theory. A projected block-coordinate ascent and successive convex approximation algorithm is developed to optimize this deterministic objective under per-device power constraints. Experiments on ModelNet10 validate the approximation and show that the proposed statistical precoder improves task-aware mode allocation and inference performance over competitive benchmarks.
\end{abstract}


\section{Introduction}

\IEEEPARstart{W}{ith} the continuing evolution toward sixth-generation (6G) networks, wireless systems are expected to serve not only as connectivity infrastructures, but also as platforms for distributed intelligence~\cite{letaief2019roadmap}. In such systems, intelligence is no longer confined to a single device or cloud server, but is distributed across sensors, edge devices, edge servers, and wireless links that jointly support decision-making~\cite{letaief2022edgeai}. This shift calls for task-oriented communication and physical-layer designs that preserve task-relevant information for downstream tasks, rather than aiming solely at bit recovery, data reconstruction, or throughput~\cite{gunduz2023beyond}.

Existing studies on task-oriented communication have focused on learning compact task-relevant representations for downstream intelligent task inference~\cite{xie2021deepsc,shao2022learning,gunduz2023beyond,shao2023multidevice}. Nevertheless, in multiple-input multiple-output (MIMO) systems, the inference performance depends critically on both the precoding of the features and the channel conditions. Therefore, recent works have started to incorporate task objectives into physical layer design. Task-oriented over-the-air computation was proposed in~\cite{wen2023task}, where the transceivers were jointly designed for a classification-oriented surrogate instead of the conventional aggregation mean-square error (MSE). Beyond over-the-air computation, an end-to-end framework for task-oriented communication over MIMO channels was developed in~\cite{cai2025end}, where feature encoders, transmit precoders, and inference models at the receiver side were jointly trained under an information-theoretic objective. More closely, the authors in \cite{cai2024taskcomm} adopt the maximal coding rate reduction (MCR$^2$) as the performance metric for multi-device edge inference system design. This metric captures the separability among different classes and only depends on the second-order statistics of the features.

Notably, prior works commonly assume that instantaneous channel state information at the transmitter (CSIT) is available. However, in practical systems, the transmitter typically acquires channel state information (CSI) via feedback from the receiver. In rapidly varying environments, accurate channel estimation becomes challenging, making instantaneous CSIT unachievable. Moreover, the precoder must be redesigned in such scenarios. Given the high computational complexity of task-oriented performance metrics, this frequent recomputation becomes prohibitive~\cite{marzetta2010noncooperative,larsson2014massive,love2008overview}. To address this issue, a more practical approach is to consider statistical CSIT for precoder design~\cite{jafar2005covariance,larsson2014massive}. In conventional MIMO communications, the ergodic capacity and outage probability have been extensively studied by leveraging tools from random matrix theory (RMT) \cite{tulino2004random, couillet2011book, hachem2008mimo, couillet2011macde, Dumount2010Rician, Wen2007}, yielding rich insights. However, the analysis and optimization of task-oriented performance metrics based on statistical CSIT still remain in their infancy.
\par
In this paper, we consider the multi-device edge inference system design with statistical CSIT. We also consider the received-feature MCR\(^2\) utility as the task-aware performance metric. Unfortunately, this metric consists of a set of coding rate reduction functions, where the feature statistics and the heterogeneous multi-device MIMO channels are coupled in each function, making it very challenging to evaluate. To this end, we first leverage RMT to derive the approximation for the metric. Based on the analytical result, we propose an optimization algorithm to obtain the task-oriented precoder.

The main contributions are summarized as follows.
\begin{itemize}
    \item [1)] For the multi-device edge inference system, we formulate a precoder optimization problem under classification tasks. In particular, we adopt the received-feature MCR\(^2\) utility as the task-aware objective and optimize the precoder under per-device power constraints, given training-set feature statistics and statistical CSIT.
    \item [2)] We derive a deterministic approximation for MCR\(^2\) utility using RMT, and the convergence rate is of order $\cO(N^{-1})$. This result is general and can be applied to the analysis of distributed and decentralized MIMO systems.
    \item [3)] The theoretical result is applied to design the precoder for maximizing the MCR\(^2\) utility. To address the non-convexity and the block structure constraint of the precoder, we propose a block coordinate ascent (BCA) method that adopts successive convex approximation (SCA) for each block.
    \item [4)] We conduct numerical experiments on ModelNet10. The experiment results demonstrate the accuracy of the theoretical results and the effectiveness of the proposed method. It is shown that the performance of the proposed precoder approaches that of the oracle in favorable regimes. We further provide diagnostic analyses that illustrate how the proposed precoder distributes transmit power among task-relevant features and channel modes.
\end{itemize}

The rest of this paper is organized as follows. 
In Section~\ref{sec:related_work}, we review related work. 
In Section~\ref{sec:system_model}, we introduce the system model and formulate the task-oriented precoding problem. 
In Section~\ref{sec:deterministic_equivalent}, we derive the deterministic approximation for the received-feature utility.
The proposed BCA-SCA algorithm for optimizing the precoder is provided in Section~\ref{sec:optimization}.
Section~\ref{sec:experiments} reports numerical experiments and  
Section~\ref{sec:conclusion} concludes the paper.

\emph{Notations:} Throughout this paper, lowercase and uppercase boldface letters denote vectors and matrices, respectively. We use \(\C^N\) and \(\C^{N\times M}\) to denote the spaces of \(N\)-dimensional complex vectors and \(N\)-by-\(M\) complex matrices, respectively. The transpose and conjugate transpose operators are denoted by \((\cdot)^T\) and \((\cdot)^H\), respectively. We write \(\|\cdot\|\) for the spectral norm of a matrix or the Euclidean norm of a vector, and \(\|\cdot\|_\infty\) for the infinity norm. The trace and determinant of a matrix \(\BA\) are denoted by \(\Tr(\BA)\) and \(\det(\BA)\), respectively. The identity matrix of size \(N\) and the all-zero vector are denoted by \(\BI_N\) and \(\vzero\), respectively. For Hermitian matrices \(\BA\) and \(\BB\), \(\BA\succeq\BB\) means that \(\BA-\BB\) is positive semidefinite. The notation \(\blkdiag(\cdot)\) denotes the block-diagonal matrix. The probability measure and expectation operator are denoted by \(\Prb[\cdot]\) and \(\E[\cdot]\), respectively. \(\CN(\Bmu,\BSigma)\) denotes the circularly symmetric complex Gaussian distribution with mean \(\Bmu\) and covariance \(\BSigma\). For a positive integer \(N\), \([N]\) denotes the set \(\{1,\ldots,N\}\). The notation \(\xrightarrow[]{a.s.}\) denotes almost sure convergence, and \(\cO(\cdot)\) denotes the standard Big-O notation. 
\section{Related Work}
\label{sec:related_work}

\subsection{Task-Oriented Precoder Design}
The design goal of task-oriented communication systems is to enhance the ability of the received signal to serve downstream tasks, rather than merely focusing on the reconstruction of the transmitted signal~\cite{strinati2021goal}. Existing works mainly consider inference tasks and focus on learning effective data representations and designing end-to-end transceivers to improve the inference performance \cite{xie2021deepsc,shao2022learning,gunduz2023beyond,shao2023multidevice}. At the physical layer, precoding directly affects the transmission of encoded features over the wireless channel, thereby influencing the performance of downstream tasks. In \cite{wen2023task}, the authors proposed a task-oriented over-the-air computation scheme where the transmit precoding and receive beamforming were jointly designed to improve the inference accuracy. In \cite{cai2025end}, an end-to-end framework for task-oriented semantic communication over MIMO channels was developed, in which feature encoders, precoders, and classifiers were jointly trained by incorporating the instantaneous CSI into the end-to-end learning pipeline. More closely related to this work, received-feature MCR\(^2\)-based precoding was proposed in~\cite{cai2024taskcomm}, where the class separability of the received representation was used as the semantic utility. However, these task-oriented precoder designs rely on instantaneous CSIT or state-dependent precoder feedback. The design of precoders for semantic communication systems based on statistical CSIT remains unexplored.

\subsection{MCR$^2$ Representation Learning}
MCR$^2$ was introduced as a geometric criterion for discriminative representation learning \cite{yu2020mcr2}. The main idea is to evaluate feature quality through a coding-rate decomposition: a desirable representation should preserve a rich global feature structure while making samples from the same class concentrate around lower-dimensional class-specific structures. Accordingly, MCR$^2$ is defined as the coding rate of the entire feature set minus the weighted sum of the coding rates of per-class feature sets, thereby encouraging both inter-class diversity and intra-class compactness. Recent work has further analyzed the geometric properties and the optimal solution of MCR$^2$, providing theoretical insight into why this criterion yields structured and discriminative representations~\cite{wang2024globalmcr2}. For task-oriented communication systems, MCR$^2$ has also been employed to characterize the received-feature class-separability for semantic precoding design~\cite{cai2024taskcomm}, which motivates its use as a task-aware performance metric in this paper.


\subsection{Analysis and Precoder Design for MIMO Communications with Statistical CSIT}

Multi-antenna technology has proven to be a promising approach for enhancing both spectral and energy efficiency \cite{telatar1999capacity,goldsmith2003capacity, vu2007linear}.  Due to the complex distributions and spatial correlations, performance analysis for MIMO systems is essential. Fortunately, RMT provides a powerful tool for handling the randomness of MIMO channels  \cite{couillet2011book}. For instance, in \cite{hachem2008mimo}, the authors derived the deterministic approximations for the ergodic mutual information and outage probability under correlated Rayleigh channels.  Subsequently, the ergodic mutual information for Rician channels was analyzed in \cite{Dumount2010Rician}. As antenna architectures evolve to offer higher flexibility, recent studies have covered more general systems, including distributed MIMO channels \cite{Junzhang2013, zhuang2025DBP}, double-scattering MIMO channels \cite{zhang2023doubleScattering}, intelligent reflecting surface (IRS)-aided MIMO channels \cite{zhang2024irsWiretap, zhuang2025}, and holographic MIMO channels \cite{zhang2024nonCenteredNonSeparable}. Accordingly, the precoder was designed based on these analyses with statistical CSIT \cite{Coui2011MIMOMAC, Xin2022}. In particular, to achieve the ergodic rate region boundary of the MIMO multiple access channel, an iterative water-filling algorithm was proposed to optimize the precoder in \cite{Coui2011MIMOMAC}. In \cite{Xin2022}, the authors considered the IRS-aided MIMO communication system, where the precoder and the phase-shift matrix are jointly optimized to minimize the outage probability. However, these designs are typically developed for capacity-oriented transmission or signal reconstruction, and their objectives are not aligned with task utility. This leaves open the problem of MIMO precoding for task-oriented feature transmission with statistical CSIT.

\par


\section{System Model and Problem Formulation}
\label{sec:system_model}

This section presents the system model for cooperative multi-device MIMO inference with statistical CSIT, as illustrated in Fig~\ref{fig:system_model}. In particular, we first describe the encoded-feature transmission model and the heterogeneous Rayleigh channel model. We then introduce the MCR\(^2\) utility for the received features and formulate the precoder design problem accordingly.

\begin{figure*}[t]
\centering
\includegraphics[width=0.9\textwidth]{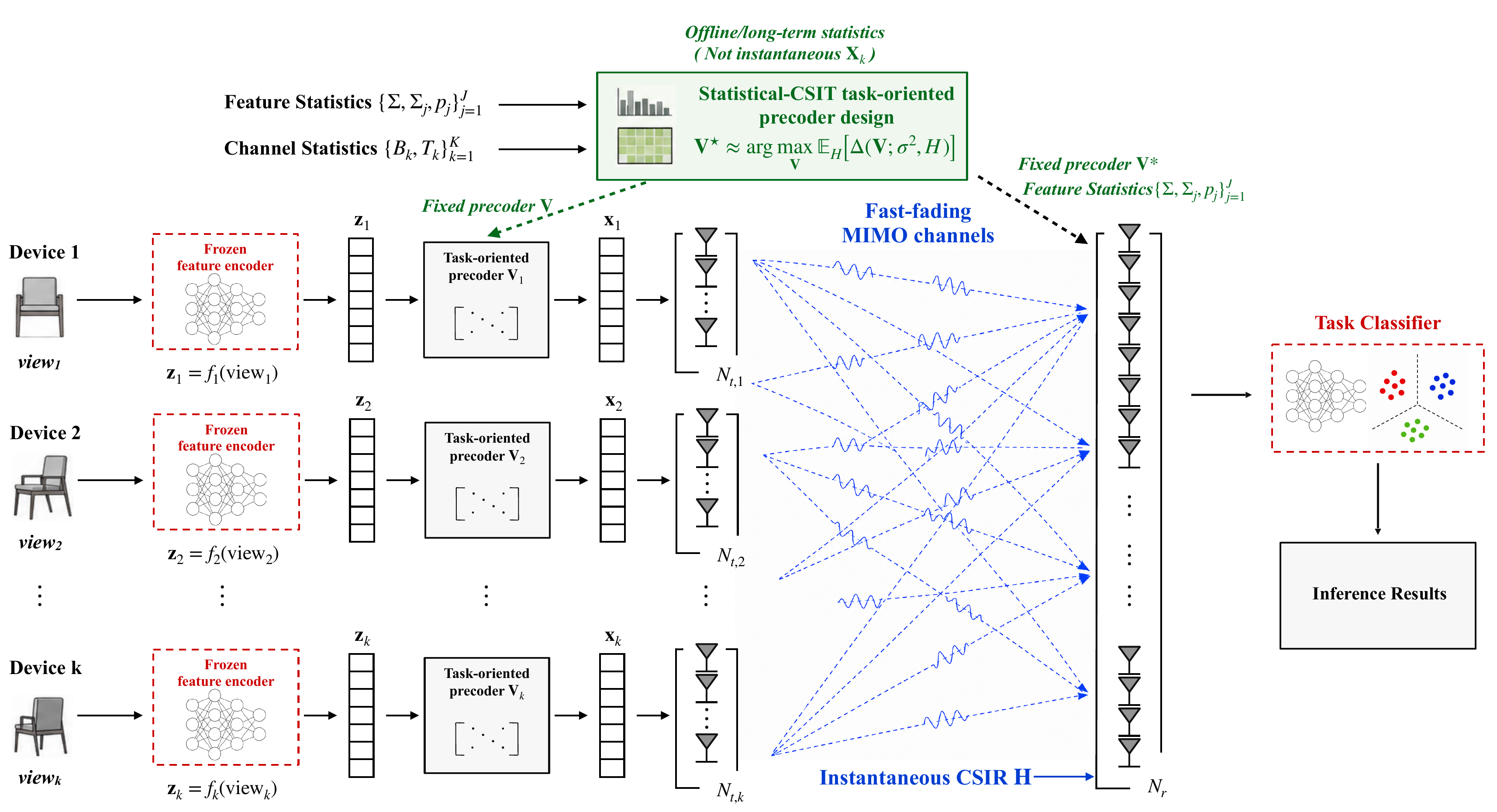}
\caption{An illustration of multi-device edge inference over MIMO systems. At the transmitter side, the precoder $\BV$ is designed based on channel statistics \(\{\BB_k,\BT_k\}_{k=1}^{K}\) and encoded training-feature statistics \(\{\BSigma,\BSigma_j,p_j\}\).
}
\label{fig:system_model}
\end{figure*}

\subsection{System Model}
\subsubsection{Signal Model}
\label{subsec:statistical_csit_signal_model}
\label{subsec:semantic_precoding_power}
\label{subsec:block_gaussian_channel}

We consider a multi-device cooperative MIMO edge-inference system with \(K\) edge devices and one edge server~\cite{cai2024taskcomm}. Device \(k\in[K]\) observes a local input \(\Bo_k\) and maps it through a fixed feature encoder \(f_k(\cdot;\theta_k)\) into a raw encoded feature vector \(\widetilde{\Bz}_k=f_k(\Bo_k;\theta_k)\in\C^{d_k}\). To allocate the transmit power to the informative signal components, the transmitted feature vector is centered as $\Bz_k = \widetilde{\Bz}_k - \E[\widetilde{\Bz}_k]$.

\par
Assume each edge device \(k \in [K]\) is equipped with \(N_{t,k}\) antennas and the edge receiver is equipped with \(N_r\) antennas. For each device \(k \in [K]\), a linear precoder \(\BV_k\in\C^{N_{t,k}\times d_k}\) is applied to \(\Bz_k\)
and the transmitted signal of device \(k\) is 
\begin{equation}
    \Bx_k
    =
    \BV_k\Bz_k
    \in\C^{N_{t,k}}.
    \label{eq:device_transmitted_signal}
\end{equation}
As a result, the average transmitted power at device $k$ is given by
\begin{equation}
    \E\left[ \|\Bx_k\|^2 \right]
    =
    \Tr \left(\BV_k\BSigma^{(kk)}\BV_k^H \right),
    \label{eq:device_average_power}
\end{equation}
where \(\BSigma^{(kk)}=\E[\Bz_k\Bz_k^H]\) denotes the covariance matrix of the local centered feature. With per-device power budget $P_k > 0$, the power constraint for precoder $\BV_k$ is given by
\begin{equation}
 \Tr \left(\BV_k\BSigma^{(kk)}\BV_k^H \right) \leq N_{t,k}P_k.
\label{eq:power_constraint}
\end{equation}
Here, the factor $N_{t, k}$ is used to normalize the transmit power and is subsequently absorbed into the channel. The details are given in \eqref{eq:block_gaussian_user_channel}.
\par
We assume the symbol-level synchronization among the devices at the edge receiver and the received signal \(\By\in\C^{N_r}\) can be expressed as
\begin{equation}
    \By
    =
    \sum_{k\in[K]}\BH_k\Bx_k+\Bn,
    \label{eq:received_signal}
\end{equation}
where \(\Bn\sim\CN(\vzero,\delta_0^2\BI_{N_r})\) denotes the additive white Gaussian noise (AWGN) at the receiver and \(\BH_k\in\C^{N_r\times N_{t,k}}\) represents the channel between device \(k\) and the receiver.  To simplify the notation, the signal model can be rewritten in a more compact form 
\begin{align}
    \By = \BH \BV \Bz + \Bn,
\end{align}
where we define $\BH = [\BH_1, \ldots, \BH_K] \in \C^{N_r \times N_k}$,  $\BV = \blkdiag(\BV_1, \ldots, \BV_K) \in \C^{N_k \times d}$, and $\Bz = [\Bz_1, \ldots, \Bz_K] \in \C^d$, with dimensions $N_{t} = \sum_{k \in [K]}N_{t, k}$ and $d = \sum_{k \in [K]} d_k$.
\par
\subsubsection{Channel Model}: Due to limited angular spread and insufficient antenna spacing, spatial correlation between the transceiver antennas is inevitable. In this paper, we consider the Kronecker correlated Rayleigh model~\cite{Hoydis2013massivemimo} for the channel:
\begin{equation}
    \BH_k
    =
    \sqrt{\beta_k} \BR_k^{1/2}\BX_k\BT_k^{1/2},
    \label{eq:block_gaussian_user_channel}
\end{equation}
where \(\BR_k\in\C^{N_r\times N_r}\) and \(\BT_k^{1/2}\in\C^{N_{t,k}\times N_{t,k}}\) are Hermitian nonnegative and represent the spatial correlation of the receive and transmit antennas, respectively. The large-scale fading coefficient is denoted by \(\beta_k\). Here, \(\BX_k\in\C^{N_r\times N_{t,k}}\) is a random matrix with i.i.d. \(\CN(0,N_{t,k}^{-1})\) entries, and \(\BX_1,\ldots,\BX_K\) are mutually independent.

Since the environment changes rapidly, it is difficult to obtain instantaneous CSI at the transmitter side. Therefore, we consider that the edge devices only know the statistical CSI \(\{\BB_k,\BT_k\}_{k=1}^{K}\) with $\BB_k = \beta_k \BR_k$ denoting the effective correlation matrix for simplicity, and the feature statistics \cite{jafar2005covariance,jorswieck2004correlation}, which are defined in the following section.

\subsubsection{MCR\(^2\) Utility}
\label{subsec:instantaneous_mcr2}

We consider a classification task and assume there are \(J\) classes. Let $\ell(\Bz) \in [J]$ denote the label of feature $\Bz$ and $p_j = \Prb[\ell(\Bz) = j]$.  The covariance matrix and the class-conditional covariance matrix of the full feature vector $\Bz$ are denoted by $ \BSigma = \E[\Bz\Bz^H]$ and $\BSigma_j = 
    \E [(\Bz-\Bmu_j)(\Bz-\Bmu_j)^H\mid \ell(\Bz)=j ]
$,
respectively, where \(\Bmu_j=\E[\Bz\mid\ell(\Bz)=j]\) represents the conditional mean. Then, we define the feature statistics as 
$
    \{\BSigma, \BSigma_1, \ldots, \BSigma_J, p_1, \ldots, p_J\},
$
which are estimated from the training set in the implementation.
\par

For receiver-side classification, we use a maximum a posteriori (MAP) classifier based on a Gaussian Mixture Model (GMM) approximation for the class-conditional received-feature distribution \cite{cai2024taskcomm, reynolds2009gaussian}. The edge receiver is assumed to know the instantaneous channel state information (CSIR) \(\mathbf H\) and the precoder \(\mathbf V\). It also has access to the training-set feature statistics needed to form the Gaussian likelihood. The label of the received feature is then inferred as
\begin{equation}
    \widehat{\ell}
    = \arg\max_{j \in [J]} p(j| \By, \BH, \BV) = 
    \arg\max_{j\in[J]}
    p_j\,p(\By\mid j,\BH, \BV).
    \label{eq:map_classifier}
\end{equation}
Directly optimizing the MAP classification rule is generally intractable. Following the received-feature MCR\(^2\) formulation in \cite{cai2024taskcomm}, we therefore use a class-separability surrogate that increases the coding rate of the received total feature distribution while penalizing the weighted coding rates of the received class-conditional feature distributions. In particular, the MCR\(^2\) surrogate utility for inference accuracy is defined as
\begin{equation}
\label{eq: mcr2 def}
    \Delta(\BV;\sigma^2, \BH)
    =
    I_{\BSigma}(\BV;\sigma^2, \BH)
    -
    \sum_{j=1}^{J}p_jI_{\BSigma_j}(\BV;\sigma^2, \BH),
\end{equation}
where
\begin{equation}
\label{eq: I bcC term}
    I_{\bcC}(\BV;\sigma^2, \BH)
    =
    \log\det\!\left(
        \BI_{N_r}
        +
        \frac{1}{\sigma^2}
        \BH\BV\bcC\BV^H\BH^H
    \right)
\end{equation}
for  \(\bcC\in\widetilde{\Sigma} = \{\BSigma, \BSigma_1,\ldots, \BSigma_J\}\). Here, the parameter $\sigma^2$ is given by $\sigma^2 = \beta / \alpha$, with \(\beta=1+\alpha\delta_0^2\) and \(\alpha=N_r/\varepsilon^2\), where \(\varepsilon>0\) represents the MCR\(^2\) coding precision. We note that the function $I_{\BSigma}$ in \eqref{eq: mcr2 def} quantifies the coding rate of the received feature vector $\By$, whereas the function $I_{\BSigma_j}$ measures the coding rate for features belonging to class $j$. 

\subsection{Problem Formulation}
\label{subsec:statistical_csit_problem}
In this paper, we aim to design the precoder that maximizes the MCR$^2$ surrogate function in \eqref{eq: mcr2 def} under the power constraint \eqref{eq:power_constraint}. To this end, we define the feasible set of the concatenated precoder matrix as 
$
    \cV= \bigotimes_{k \in [K]} \cV_k,
$
where
\begin{equation}
    \cV_k =
    \left\{
    \BV_k \in \C^{N_{t, k} \times d_k}:
    \Tr(\BV_k\BSigma^{(kk)}\BV_k^H)\le N_{t, k} P_k
    \right\}.
\end{equation}
Then, the precoding optimization problem is formulated as
\begin{align}
   \cP_1: \max_{\BV\in\cV}\quad
    &
    \E\!\left[\Delta(\BV;\sigma^2, \BH)\right],
    \label{eq:statistical_csit_objective}
\end{align}
where the expectation is taken over the small-scale fading components \(\{\BX_k\}_{k=1}^{K}\). 
Problem~\eqref{eq:statistical_csit_objective} is challenging to solve for two reasons. First, the objective function involves an expectation over log-determinant terms, which makes direct evaluation computationally prohibitive. Second, the optimization problem is non-convex due to the difference of these log-determinant terms. To overcome these difficulties, in the following, we first leverage RMT to derive the closed-form deterministic approximation for the objective function, which is in closed-form. Then, we propose a block coordinate ascent (BCA) algorithm to optimize the precoder.
{ 
\section{Deterministic Approximation for the MCR$^2$ Objective }
\label{sec:deterministic_equivalent}
In this section, we first present a general result establishing the convergence of the resolvent for the covariance matrix of sums of Gaussian random matrices. Based on this theoretical result, we then obtain the desired deterministic approximation.
\subsection{Theoretical Results}
The surrogate objective $\Delta$  comprises a linear combination of log-determinant terms whose behavior is governed by the eigenvalue distribution of the channel matrix $\BH$. To study this distribution, a fundamental tool is to investigate the resolvent for $\BH$ \cite{hachem2007deterministic}. The following theorem provides the deterministic approximation for the resolvent of a broad class of random matrices that includes $\BH$ as a special case.
\begin{theorem}
\label{thm: convergence resolvent}
    Let $\{ \BZ_k\}_{k=1}^K$ be a collection of $K$ random matrices with size $N_r \times N_t$. For each $k \in [K]$, suppose $\BZ_k = \BA_k \BX_k \widetilde{\BA}_k$, where $\BA_k \in \C^{N_r \times L_k}$ and $\widetilde{\BA}_k \in \C^{M_k \times N_t}$ are deterministic, and $\BX_k \in \C^{L_k \times M_k}$ is a random matrix with i.i.d. $\cC\cN(0, M_k^{-1})$ elements.  Define $\BOmega_k = \BA_k\BA_k^H$ and $\widetilde{\BOmega}_k = \widetilde{\BA}_k^H \widetilde{\BA}_k$. We impose the following assumptions: 
    \begin{itemize}
        \item [i)] $K$ is a given constant. Let $d_k = L_k / N_r$, $\widetilde{d}_k = M_k / N_r $ for $k \in [K]$. Dimensions $L_k, M_k, N_{r}$, and $N_t$ approach infinity such that 
        $
        0 < \lim\inf \min_{k \in [K]} \{d_k, \widetilde{d}_k\} \leq \lim\sup \max_{k \in [K]} \{d_k, \widetilde{d}_k\} < \infty$ and $
        0 < \liminf N_r / N_t \leq \limsup N_r / N_t < \infty$.
        \item [ii)] $\limsup \max_{k \in [K]}\{ \lnorm{\BOmega_k}, \lnorm{\widetilde{\BOmega}_k} \} < \infty$.
        \item [iii)] $\liminf \min_{k \in [K]} \{ \frac{1}{N_r} \Tr \BOmega_k, \frac{1}{N_t} \Tr \widetilde{\BOmega}_k \} > 0$.
    \end{itemize}
    Let $\BZ = \sum_{k=1}^K \BZ_k$ and define its resolvent as
    \begin{align}
    \BQ(z) =  \left(-z\BI_{N_r} + \BZ\BZ^H\right)^{-1}, ~~ z < 0.
    \end{align}
    Then, for any deterministic matrix $\bcR \in \C^{N_r \times N_r}$ with uniformly bounded norm, i.e., $\limsup \lnorm{\bcR} < \infty$, we have 
    \begin{align}
    \label{eq: poly bound}
        \Tr \bcR\E \BQ(z) - \Tr \bcR \BTheta(z) = \cO\left( \frac{\mathsf{P}_5(|z|^{-1})}{N_r |z|^4}\right),
    \end{align}
    where $\mathsf{P}_5(|z|)$ is a polynomial of degree at most $5$ with positive coefficients and $\BTheta(z)$ is defined by the solution of the following self-consistent system of equations:
    \begin{subequations}
    \label{eq: de}
    \begin{align}
        \delta_k(z) &= \frac{1}{M_k} \Tr  \BOmega_k \BTheta(z), \quad k \in [K], \\
        \widetilde\delta_k(z) &= \frac{1}{M_k} \Tr \widetilde{\BOmega}_k \widetilde{\BTheta}(z), \quad k \in [K], \\
        \BTheta(z) &= -\frac{1}{z}\Bigg[\BI_{N_r} + \sum_{k \in [K]} \widetilde{\delta}_k(z)\BOmega_k \Bigg]^{-1}, \\
        \widetilde{\BTheta}(z) &= -\frac{1}{z} \Bigg[\BI_{N_t} + \sum_{k \in [K]} \delta_k(z) \widetilde{\BOmega}_k \Bigg]^{-1},
    \end{align}
    \end{subequations}
    with \(\delta_k(z)>0\) and \(\widetilde{\delta}_k(z)>0\) for all \(k\in[K]\) and \(z<0\).
\end{theorem}
\textit{Proof:} The proof of Theorem \ref{thm: convergence resolvent} is given in Appendix \ref{app: convergence resolvent}. \QED
\par
We note that the positive solution to the fixed-point equations \eqref{eq: de} can be proved to exist and be unique. A numerical procedure for computing this solution is summarized in Algorithm \ref{alg:de fixed-point}. We provide the following technical remarks for a more comprehensive understanding of the theoretical results in Theorem~\ref{thm: convergence resolvent}.
\begin{remark}
    The polynomial bound in \eqref{eq: poly bound} is derived to ensure that this error is integrable over an unbounded region of $z < 0$.
\end{remark}
\begin{remark}
    Theorem~\ref{thm: convergence resolvent} can be used for approximating the linear functional of the eigenvalues of $\BZ\BZ^H$. Denote the ordered eigenvalues of $\BZ\BZ^H$ as $0 \leq \lambda_1 \leq \lambda_2 \leq \cdots \leq \lambda_{N_r}$ and the empirical spectral distribution as 
   $
    F^{\BZ\BZ^H}(x) = \frac{1}{N_r}\sum_{i=1}^{N_r} \mathbb{I}_{\{\lambda_i \leq x\}}.
    $
    According to the properties of the Stieltjes transform, there exists a probability measure $\overline{F}$ such that
    $
        \frac{1}{N_r} \Tr \boldsymbol{\Theta}(z) = \int \frac{\overline{F}(dx)}{x - z}.
    $
    Using Theorem \ref{thm: convergence resolvent}, one can prove that
    \begin{align}
        \int_{\R} f(x) \left[ F^{\BZ\BZ^H}(dx) - \overline{F}(dx) \right] \xrightarrow[]{a.s.} 0,
    \end{align}
    for any bounded continuous function $f$. In particular, by setting $f(x) = \log(1 + x/\sigma^2)$, we obtain the approximation for the MCR$^2$ utility.
\end{remark}

\begin{algorithm}[t]
\caption{Fixed-Point Algorithm for Solving \eqref{eq: de}}
\label{alg:de fixed-point}
\begin{algorithmic}[1]
\Require  parameter $z < 0$, correlation matrices $\{\BA_k\}_{k=1}^K$ and $\{\widetilde{\mathbf A}_k\}_{k=1}^K$, initial points $\delta_{k}^{(0)} > 0$ and $\widetilde{\delta}_k^{(0)} > 0$, for $k \in [K]$, tolerance $\epsilon$.
\Ensure Solutions $\{\delta_{k}\}_{k=1}^K$ and $\{\widetilde{\delta}_k\}_{k=1}^K$. 
\State Set \(t\gets 1\) 
\Repeat
    \State  Set $\BTheta^{(t)} \gets  -\frac{1}{z} (\BI_{N_r} + \sum_{k \in [K]} \widetilde{\delta}_k^{(t-1)} \BOmega_k)^{-1}$.
    \State Set $\widetilde{\BTheta}^{(t)} \gets  -\frac{1}{z} (\BI_{N_t} + \sum_{k \in [K]}\delta_k^{(t-1)} \widetilde{\BOmega}_k)^{-1}$.
    \For{\(k=1,\ldots,K\)}
        \State Compute $\delta^{(t)}_k \gets \frac{1}{M_k} \Tr \BOmega_k \BTheta^{(t)}$.
        \State Compute $\widetilde{\delta}^{(t)}_k \gets \frac{1}{M_k} \Tr \widetilde{\BOmega}_k \widetilde{\BTheta}^{(t)}$.
    \EndFor
    \State Set $\epsilon_t \gets \max_{k \in [K]}|\delta_k^{(t)} - \delta_k^{(t-1)}|$  and \(t\leftarrow t+1\).
\Until{ $\epsilon_t \leq \epsilon$.}
\State Return $\{\delta_{k}^{(t-1)}\}_{k=1}^K$ and $\{\widetilde{\delta}_k^{(t-1)}\}_{k=1}^K$.
\end{algorithmic}
\end{algorithm}

\subsection{Approximation for the MCR$^2$ utility}

In this section, we apply Theorem 1 to obtain the deterministic approximation for the optimization objective in \eqref{eq:statistical_csit_objective}.
To this end, we introduce 
\begin{align}
    \bcT_k = \blkdiag(\mathbf{0}_{N_{t, 1} \times N_{t, 1}}, \ldots, \BT_k, \ldots, \mathbf{0}_{N_{t, K}\times N_{t, K}}),
\end{align}
where the only non-zero block is the $k$-th diagonal block, and all other blocks are zero matrices.
To facilitate the analysis, we make the following assumptions.
\begin{assumption}
\label{ass: dim}
 The numbers \(K\) and $J$ are given. Dimensions $N_r$ and $N_{t, k}$ approach infinity with
    \[
        0<\liminf_{N_r}  \frac{N_{t,k}}{N_r}
        \le
        \limsup_{N_r}  \frac{N_{t,k}}{N_r}<\infty, ~~ k \in [K].
    \]
\end{assumption}
\begin{assumption}
\label{ass: sp norm}
     The spectral norms of $\BB_k$, $\BT_k$, and $\BP_{\bcC}$ are uniformly bounded, i.e., 
    \begin{align}
        &\limsup_{N_r}  \max_{k \in [K]} \left\{ \| \BB_k \|, \|\BT_k\| \right\} < \infty, \\
        &\limsup_{N_r} \max_{\bcC \in \widetilde{\Sigma} } \| \BP_{\bcC}\|< \infty. 
    \end{align}
\end{assumption}
\begin{assumption}
\label{ass: trace}
     The normalized traces of the correlation matrices are uniformly bounded away from $0$. In particular, 
    \begin{align}
        &\liminf_{N_r} \min_{k \in [K]} \frac{1}{N_r} \mathrm{Tr} \BB_k > 0, \\
        &\liminf_{N_r} \min_{\bcC \in \widetilde{\Sigma}, k \in [K]}\frac{1}{N_r} \Tr \bcT_k \BP_{\bcC} > 0.
    \end{align} 
\end{assumption}
Assumption \ref{ass: dim} is the high-dimensional system assumption, i.e., the number of antennas goes to infinity at the same ratio, which is common in the study of RMT \cite{couillet2011book}. In the following, for notational convenience, we use $N_r \to \infty$ to denote this asymptotic regime. Assumptions \ref{ass: sp norm} and \ref{ass: trace} ensure that the spectral norm of the spatial correlation matrices is bounded and prevent antenna imbalance such that the rank of the spatial correlation matrix is very small.
\par
The following theorem provides the deterministic approximation for the terms in the MCR$^2$ utility.
\begin{theorem}
\label{thm: EMI approx}
Assume Assumptions \ref{ass: dim}-\ref{ass: trace} hold. Then, we have, for $\bcC \in \widetilde{\Sigma}$,
\begin{align}
\E \left[I_{\bcC}(\BV; \sigma^2, \BH) \right] &= \overline{I}_{\bcC}(\BV; \sigma^2) + \cO(N_r^{-1})
\end{align}
where 
\begin{align*}
    \overline{I}_{\bcC}(\BV; \sigma^2) &= \log\det\Big(\BI_{N_r} + \sum_{k \in [K]} \widetilde{\delta}_{\bcC, k}(-\sigma^2) \BB_k \Big) \\
    & + \log\det\Big( \BI_{N_t} +  \BD_{\bcC, \delta}(-\sigma^2) \BP_{\bcC} \Big) \\
    &- \sum_{k \in [K]} \sigma^2\delta_{\bcC, k}(- \sigma^2) \widetilde{\delta}_{\bcC, k}(-\sigma^2), \numberthis
\end{align*}
where  $\BD_{\bcC, \delta}(-\sigma^2) = \blkdiag(\delta_{\bcC, k}(-\sigma^2) \BT_k; k \in [K])$ and $\bcC \in \widetilde{\Sigma}$. Here, $\{ \delta_{\bcC, k}(-\sigma^2), \widetilde{\delta}_{\bcC, k}(-\sigma^2) \}_{k=1}^K$ is the positive solution for the following system of equations 
\begin{subequations}
\label{eq: de MCR2 equation}
\begin{align}
    \delta_{\bcC, k}(-\sigma^2) = \frac{1}{\sigma^2 N_{t, k}} \Tr \BB_k \Big( \BI_{N_r} + \sum_{l \in [K]} \widetilde{\delta}_{\bcC, l}(-\sigma^2) \BB_{l} \Big)^{-1}, \\
    \widetilde{\delta}_{\bcC, k}(-\sigma^2) = \frac{1}{\sigma^2 N_{t, k}} \Tr \bcT_k \BP_{\bcC}\Big( \BI_{N_t} + \BD_{\bcC, \delta}(-\sigma^2) \BP_{\bcC} \Big)^{-1},
\end{align}
\end{subequations}
for $k \in [K]$ and $\bcC \in \widetilde{\Sigma}$.
\end{theorem}
}
\textit{Proof:} The proof for Theorem \ref{thm: EMI approx} is given in Appendix \ref{app: proof of thm EMI approx}. \QED
\par
By Theorem~\ref{thm: EMI approx} and Assumption~\ref{ass: dim} that the number of classes $J$ is given, we immediately obtain the deterministic approximation for the MCR$^2$ utility with the corresponding convergence rate 
\begin{align}
\E[\Delta(\BV; \sigma^2, \BH)] =  \overline{\Delta}(\BV; \sigma^2)+ \cO(N_r^{-1}), \label{eq: de MI} 
\end{align}
where 
$
    \overline{\Delta}(\BV; \sigma^2) = \overline{I}_{\BSigma}(\BV; \sigma^2) - \sum_{j=1}^J p_j \overline{I}_{\BSigma_j}(\BV; \sigma^2)
$.
\begin{remark}
    Theorem \ref{thm: EMI approx} has wide applicability in the analysis of MIMO systems. For example, setting $K = 1$, Theorem \ref{thm: EMI approx} reduces to \cite[Theorem 1]{hachem2008mimo} and can be used to approximate the ergodic mutual information of point-to-point MIMO systems. For a general $K$, Theorem~\ref{thm: EMI approx} has potential applications in distributed or decentralized MIMO systems \cite{zhuang2025DBP, Junzhang2013}.
\end{remark}

\section{Statistical-CSIT Semantic Precoder Design}
\label{sec:optimization}

In this section, we optimize the precoder by utilizing the deterministic approximation in Theorem~\ref{thm: EMI approx}. With \(\overline{\Delta}(\BV;\sigma^2)\) defined in \eqref{eq: de MI}, the optimization problem $\cP_1$ in \eqref{eq:statistical_csit_objective} can be formulated as
\begin{equation}
    \begin{aligned}
    \cP_2:& \quad
    \max_{\BV  \in \cV}
    \overline{\Delta}(\BV;\sigma^2).
    \end{aligned}
    \label{eq:fp-semantic-optimization-objective}
\end{equation}
However, the above problem is still non-convex because \(\overline{\Delta}\) involves a weighted difference of log determinant terms and each \(\overline{I}_{\bcC}\) depends on $\BV$ implicitly through the fixed-point equation \eqref{eq: de MCR2 equation}. Fortunately, the derivative of the objective function with respect to \(\BV\) can be given in a closed form. Therefore, together with the product structure of the feasible set $\cV$, problem \(\cP_2\) can be solved by the BCA method.  In particular, at iteration $t$, for each block $k \in [K]$, we solve the following subproblem
\begin{align}
     \max_{\BV_k \in \cV_k} \overline{\Delta}(\BV_1^{(t)},  \ldots, \BV_{k-1}^{(t)}, \BV_k, \BV_{k+1}^{(t-1)}, \ldots; \sigma^2)
     \label{eq: BCA-SCA}
\end{align}
using the SCA and line-search method. Here, we write $\overline{\Delta}(\BV; \sigma^2) = \overline{\Delta}(\BV_1, \BV_2, \ldots, \BV_K; \sigma^2)$.
\subsection{BCA-SCA Method}


To start with, we first compute the derivative of the objective function with respect to $\BV$. Under the Wirtinger convention
$
    \mathrm{d}f(\Bx)
    =
    2\operatorname{Re}\Tr\!
    \left(\nabla_{\Bx^*}f(\Bx)\right)^H\mathrm{d}\Bx
$ and using $\partial \log\det (\BX) = \Tr \BX^{-1}\partial \BX$, the gradient for $\overline{I}_{\bcC}$ can be given by
\begin{align}
    \nabla_{\BV^*}\overline{I}_{\bcC}(\BV;\sigma^2)
    =
    \widetilde{\BTheta}_{\bcC}\BD_{\bcC,\delta}\BV\bcC,
    \label{eq:envelope-p-gradient}
\end{align}
where \(\widetilde{\BTheta}_{\bcC}=(\BI_{N_t}+\BD_{\bcC,\delta}\BP_{\bcC}(\BV))^{-1}\), and the dependence of \(\BD_{\bcC,\delta}\) on \(\sigma^2\) and \(\BV\) is omitted for brevity. Hence, by the definition in \eqref{eq: de MCR2 equation}, we have
\begin{equation}
    \nabla_{\BV^*}\overline{\Delta}(\BV;\sigma^2)
    =
    \widetilde{\BTheta}_{\BSigma}\BD_{\BSigma,\delta}\BV\BSigma
    -
    \sum_{j=1}^{J}p_j
    \widetilde{\BTheta}_{\BSigma_j}\BD_{\BSigma_j,\delta}\BV\BSigma_j .
    \label{eq:fp-v-gradient}
\end{equation}
The quantity in \eqref{eq:fp-v-gradient} is not used to perform an unconstrained joint update due to the block diagonal structure of $\BV$. Instead, it supplies the local ascent direction. 
In particular, the gradient with respect to $\BV_k$ is given by
\begin{align}
\BG_k(\BV) = \nabla_{\BV_k^*} \overline{\Delta}(\BV;\sigma^2) = [\nabla_{\BV^*}\overline{\Delta}(\BV;\sigma^2)]_{[\Btau_k, \Beta_k]},
\label{eq: derivative iter Vk}
\end{align}
where the index vectors $\Btau_k$ and $\Beta_k$ are defined as
\begin{align*}
    \Btau_k &= \Bigg[
        \sum_{l \in [k-1]} N_{t, l} + 1, \sum_{l \in [k-1]} N_{t, l} + 2, \ldots,\sum_{l \in [k]} N_{t, l}
    \Bigg]^\top, \\
    \Beta_k &= \Bigg[
        \sum_{l \in [k-1]} d_{l} + 1, \sum_{l \in [k-1]} d_{l} + 2, \ldots,\sum_{l \in [k]} d_l
    \Bigg]^\top. \numberthis
\end{align*}
By \eqref{eq: derivative iter Vk}, at the current point $\BV = \blkdiag(\BV_l; l \in [K]) \in \C^{N_t \times d}$, for each block $k \in [K]$, we can derive the first-order proximal SCA surrogate for the objective in \eqref{eq: BCA-SCA} as 
\begin{equation}
    \begin{aligned}
    \widetilde{\overline{\Delta}}_k(\BF_k|\BV,\mu)
    &=
    \rm{Constant}\\
    &+2\operatorname{Re}\Tr\!\left[
    \BG_k(\BV)^H
    \left(\BF_k-\BV_{k}\right)
    \right]\\
    &-\frac{1}{\mu}\|\BF_k-\BV_{{k}}\|_F^2,
    \end{aligned}
    \label{eq:sca-local-model}
\end{equation}
where \(\mu>0\) controls the size of the local trust region. Hence, at iteration $t$, for block $k \in [K]$, the corresponding SCA subproblem can be formulated as 
\begin{equation}
    \begin{aligned}
 \cP^t_{2, k}: \max_{\BF_k \in \cV_k}\quad
    &\widetilde{\overline{\Delta}}_k(\BF_k|\BV,\mu) \big|_{\BV = \BV^{(t, k)}},
    \end{aligned}
    \label{eq:sca-block-subproblem}
\end{equation}
where we define 
\begin{align}
\BV^{(t, k)} = \blkdiag(\BV_1^{(t)}, \ldots, \BV^{(t)}_{k-1}, \BV_k^{(t-1)}, \ldots, \BV^{(t-1)}_K).
\end{align}
One can verify that $\cP_{2,k}^t$ is convex. The following lemma demonstrates its optimal solution.
\begin{lemma}
\label{lem: opt solution SCA}
For a given feasible precoder \(\mathbf{V}\) and stepsize \(\mu>0\), define
$
    \mathbf{Z}_k
    =
    \mathbf{V}_k+\mu\mathbf{G}_k(\mathbf{V}).
    \label{eq:sca-projection-center}
$
The optimal solution of
$
    \max_{\mathbf{F}_k\in\mathcal{V}_k}
    \widetilde{\overline{\Delta}}_k(\mathbf{F}_k|\mathbf{V},\mu)
$
is given by
\begin{equation}
    \mathbf{F}_k^{\star}
    =
    \mathbf{Z}_k
    \left(
        \mathbf{I}_{d_k}
        +
        \lambda_k\boldsymbol{\Sigma}^{(kk)}
    \right)^{-1},
    \label{eq:sca-projection-solution}
\end{equation}
where \(\lambda_k\ge0\) satisfies
\begin{equation}
    \lambda_k
    \left(
        \operatorname{Tr}
        \left[
            \mathbf{F}_k^{\star}
            \boldsymbol{\Sigma}^{(kk)}
            (\mathbf{F}_k^{\star})^H
        \right]
        -
        N_{t,k}P_k
    \right)
    =
    0.
    \label{eq:sca-projection-complementary}
\end{equation}
If \(\mathbf{Z}_k\) is feasible, then \(\lambda_k=0\) and \(\mathbf{F}_k^{\star}=\mathbf{Z}_k\). Otherwise, \(\lambda_k>0\) is chosen such that
\begin{equation}
    \operatorname{Tr}
    \left[
        \mathbf{F}_k^{\star}
        \boldsymbol{\Sigma}^{(kk)}
        (\mathbf{F}_k^{\star})^H
    \right]
    =
    N_{t,k}P_k.
    \label{eq:sca-projection-active}
\end{equation}
\end{lemma}
\textit{Proof:}
The optimal solution can be obtained by completing the square in \eqref{eq:sca-local-model} and applying the KKT conditions. The details are omitted for brevity.  \QED
\par
According to Lemma \ref{lem: opt solution SCA}, at iteration $t$, $\BV_k^{(t)}$ for $k \in [K]$ can be obtained as
\begin{align}
    \BV_k^{(t)} = \arg\max_{\BF_k \in \cV_k}\quad
    &\widetilde{\overline{\Delta}}_k(\BF_k|\BV,\mu) \big|_{\BV = \BV^{(t, k)}}. \label{eq: local update Vk}
\end{align}
Note that the above step is generated from a local approximation, while the objective in \(\mathcal{P}_2\) is non-convex. Hence, the candidate block is accepted after evaluating the original objective. Let \(\BV_{+}^{(t, k)}\) be the precoder obtained by replacing only block $\BV_k^{(t)}$ with \(\BV_k^{(t+1)}\). The backtracking line search chooses \(\mu=\mu_0,\tau\mu_0,\tau^2\mu_0,\ldots\), with \(\tau\in(0,1)\), until
\begin{equation}
    \overline{\Delta}(\BV_{+}^{(t, k)};\sigma^2)
    \ge
    \overline{\Delta}(\BV^{(t, k)};\sigma^2)
    +
    \frac{c_{\rm ls}}{\mu}
    \|\BV_k^{(t)}-\BV_k^{(t-1)}\|_F^2,
    \label{eq:line-search-condition}
\end{equation}
where \(c_{\rm ls}\ge0\) represents the line-search constant.  The detailed procedure for solving $\cP_2$ is summarized in Algorithm \ref{alg:statistical-rmt-mcr2}.

\begin{algorithm}[t]
\caption{BCA-SCA Algorithm for MCR$^2$ Utility Maximization by Optimizing Precoder Matrix $\BV$}
\label{alg:statistical-rmt-mcr2}
\begin{algorithmic}[1]
\Require Channel statistics \(\{\BB_k,\BT_k\}_{k=1}^{K}\), feature statistics \(\{\BSigma,\BSigma_j,p_j\}_{j=1}^{J}\), parameter $\sigma^2$,  initial feasible precoder \(\BV^{(0)}\), initial stepsize \(\mu_0\), minimum allowed stepsize $\mu_{\rm m}$, backtracking factor \(\tau\in(0, 1)\), line-search constant \(c_{\rm ls}\ge0\), tolerance \(\epsilon_{\rm opt}\)
\Ensure A feasible precoder \(\BV\)
\State Set \(t\gets 1\) and $\BV_{\rm cur} \gets \BV^{(0)}$.
\Repeat
    \State Set \(a_{\rm old} \gets \overline{\Delta}(\BV_{\rm cur};\sigma^2)\).
    \For{\(k=1,\ldots,K\)}
    \State Set $\BV^{(t, k)} \gets \BV_{\rm cur}$.
        \State Compute \(\BG_k(\BV^{(t, k)})\) by \eqref{eq: derivative iter Vk}.
        \State Set \(\mu \gets \mu_0\) and \(\mathrm{accepted} \gets \mathrm{false}\).
        \Repeat
            \State Compute \(\BV_k^{(t)}\) by \eqref{eq: local update Vk}.
            \State Let \(\BV_{+}^{(t, k)}\) be \(\BV^{(t, k)}\) with block \(k\) replaced by \(\BV_k^{(t)}\).
            \State Evaluate \(\overline{\Delta}(\BV_{+}^{(t, k)};\sigma^2)\) by \eqref{eq: de MI}.
            \If{\eqref{eq:line-search-condition} holds}
                \State Set \(\BV_{\rm cur} \gets \BV_{+}^{(t, k)}\) and \(\mathrm{accepted}\gets \mathrm{true}\).
            \Else
                \State Set \(\mu \gets \tau\mu\).
            \EndIf
        \Until{\(\mathrm{accepted}=\mathrm{true}\) or \(\mu \leq \mu_{\rm m}\)}.
    \EndFor
    \State Set \(a_{\rm new}\gets \overline{\Delta}(\BV_{\rm cur};\sigma^2)\).
    \State Set \(t\leftarrow t+1\).
\Until{\(|a_{\rm new}-a_{\rm old}|\le\epsilon_{\rm opt}(1+|a_{\rm old}|)\)}
\State Return \(\BV_{\rm cur}\).
\end{algorithmic}
\end{algorithm}

\subsection{Computational Complexity}
\label{subsec:optimization-complexity}

The computational burden of Algorithm~\ref{alg:statistical-rmt-mcr2} is dominated by repeated evaluations of the objective and its gradient, which requires solving the fixed-point system in \eqref{eq: de MCR2 equation} and then computing the corresponding log-determinant and matrix inverses. Assume \(L_{\rm fp}\) iterations for computing the solution for \eqref{eq: de MCR2 equation} are required. Then, the computational cost for a single gradient evaluation \eqref{eq:fp-v-gradient} is given by
\begin{equation}
    \mathcal{O}\!\left((J+1)(L_{\rm fp}+1)(N_r^3+N_t^3)\right).
\end{equation}
With an average of \(L_{\rm ls}\) line-search trials per block, the dominant computational cost for a single BCA-SCA sweep over the \(K\) precoder blocks is given by
\begin{equation}
    \mathcal{O}\!\left(K(J+1)(1+L_{\rm ls})(L_{\rm fp}+1)(N_r^3+N_t^3)\right).
\end{equation}
Solving for the local surrogate optimum \eqref{eq: local update Vk} involves inverting the local block covariance \(\BSigma^{(kk)}\) and performing a scalar search for the multiplier \(\lambda_k\). If \(L_{\rm proj}\) scalar-search iterations are used, the additional cost per sweep is \(\mathcal{O}(L_{\rm proj}\sum_{k=1}^{K}d_k^3)\), which is of lower order when the feature dimensions $d_k$ are moderate compared with the number of the receive and transmit antennas.
\section{Numerical Experiments}
\label{sec:experiments}
\begin{table*}[t]
\centering
\caption{Controlled feature-role diagnostics on ModelNet10. Probe accuracies are percentages; trace ratio denotes \(\operatorname{Tr}\BSigma_u/\operatorname{Tr}\BSigma_s\), and \(\kappa\) denotes covariance condition number.}
\label{tab:feature_probe_diagnostics}
\scriptsize
\begin{tabular}{lrrrrrrrr}
\toprule
Encoder & \(\Bu\)-probe acc. & \(\Bs\)-probe acc. & \([\Bu,\Bs]\) acc. & \(\operatorname{Tr}\BSigma_u\) & \(\operatorname{Tr}\BSigma_s\) & Trace ratio & \(\kappa_u\) & \(\kappa_s\) \\
\midrule
Encoder 1 & 48.46 & 95.93 & 96.92 & 2019.54 & 303.72 & 6.65 & 20.53 & 15.75 \\
Encoder 2 & 40.64 & 96.04 & 95.93 & 899.66 & 165.14 & 5.45 & 13.84 & 20.80 \\
Encoder 3 & 55.73 & 95.15 & 95.59 & 1160.96 & 210.69 & 5.51 & 25.46 & 20.34 \\
Encoder 4 & 41.19 & 97.69 & 97.47 & 2140.37 & 313.82 & 6.82 & 24.21 & 17.12 \\
Encoder 5 & 56.61 & 96.15 & 95.70 & 1008.38 & 167.98 & 6.00 & 40.79 & 23.76 \\
\midrule
Mean \(\pm\) std & \(48.52\pm7.64\) & \(96.19\pm0.92\) & \(96.32\pm0.83\) & \(1445.78\pm587.87\) & \(232.27\pm72.22\) & \(6.09\pm0.63\) & \(24.96\pm9.93\) & \(19.55\pm3.17\) \\
\bottomrule
\end{tabular}
\end{table*}

\begin{table*}[t]
\centering
\caption{Channel information and design principles of the proposed method, baselines, and oracle method.}
\label{tab:baseline_information_sets}
\renewcommand{\arraystretch}{1.25}
\resizebox{\textwidth}{!}{%
\begin{tabular}{>{\centering\arraybackslash}m{0.20\textwidth}|>{\centering\arraybackslash}m{0.25\textwidth}|>{\centering\arraybackslash}m{0.16\textwidth}|>{\centering\arraybackslash}m{0.35\textwidth}}
\toprule
Method & Channel information & Depends on realized \(\BH\)? & Design principle \\
\midrule
Proposed BCA-SCA & Long-term statistics \(\{\BB_k,\BT_k\}_{k=1}^{K}\) & No & Deterministic approximation for the MCR\(^2\) surrogate \\
\hline
Identity & None & No & Fixed block identity map \\
\hline
Random & None & No & Feasible block-random precoders \\
\hline
Maximum-ER Eigenvalue SAA & Statistics and training-channel samples & No & Capacity over \(\BT_k\)-eigenmode powers \\
\hline
Maximum-ER Cov. SAA & Training-channel samples & No & SAA capacity over full transmit covariances \\
\hline
LMMSE SAA & Training-channel samples & No & Sample-average feature reconstruction \\
\hline
Maximum-ER Water-Filling & Long-term statistics \(\{\BB_k,\BT_k\}_{k=1}^{K}\) & No & Deterministic-equivalent water filling \\
\hline
Maximum-MCR\(^2\) Oracle & Realized channel \(\BH\) & Yes & Instantaneous MCR\(^2\) utility reference \\
\bottomrule
\end{tabular}%
}
\end{table*}
This section presents numerical experiments to validate the proposed semantic precoding optimization framework. The experiments are designed to examine both the accuracy of the deterministic approximation and the behavior of the resulting precoder under controlled task-oriented transmission settings.

\subsection{Experiment Setup}
\label{subsec:experiment_setup}

We consider the classification task on the ModelNet10~\cite{wu20153d} dataset in the multi-device edge inference system. In particular, the dataset contains \(J=10\) classes of computer-aided design objects, with each object represented by 12 rendered views. The number of samples per class is imbalanced, ranging from 106 to 889 training samples and from 50 to 100 testing samples. Following the multi-device feature interface in the task-oriented communication baseline, we select three views for each object as the inputs. In the experiments, the VGG11 backbone~\cite{Simonyan15} is adopted as the feature encoder $f_{k}(\cdot; \theta_k)$, with its parameters frozen in advance. The dimension of the feature vector is set to $d_k = 4$ for each device. The extracted training and testing feature sets contain 3991 and 908 samples, respectively, and the data statistics \(\{\BSigma,\BSigma_j,p_j\}\) are computed from the training feature set.
\par
The system dimensions are set as $K = 3$ edge devices, \(N_{t,k}=4\) transmit antennas per device, and \(N_r=8\) receive antennas.  The power budget for each edge device is \(P_k=1\)W and the signal-to-noise ratio (SNR) is defined as \(\mathrm{SNR}=1/\delta_0^2\). We consider the standard exponential correlation model \cite{Wagner2012Seb}, where the spatial correlation matrices are given by $[\BR_k]_{i, j} = r_k^{\abs{i-j}}$ and $[\BT_k]_{i, j} = t_k^{\abs{i-j}}$, with \((r_1,r_2,r_3)=(0.15,0.55,0.85)\) and \((t_1,t_2,t_3)=(0.10,0.50,0.80)\), respectively. The large-scale fading coefficients and the MCR$^2$ coding rate in \eqref{eq: mcr2 def} are set to \((\beta_1,\beta_2,\beta_3)=(0.4,1.0,0.8)\) and  \(\epsilon=10^{-3}\), respectively.

\subsection{Controlled Feature-Role Construction}
\label{subsec:feature_role_setup}
This subsection describes the construction of the controlled feature-role that will be employed in the following experiments.
The purpose of this construction is to diagnose how different precoding methods allocate transmission power to feature components with different task relevance. Specifically, by slightly modifying the loss function in the training process of encoder, we obtain two feature components with contrasting roles: the \textit{nuisance} component has a larger energy but is less informative, whereas the \textit{semantic} component has a smaller energy but is more informative.
Following the notation in Section~\ref{sec:system_model}, we write the transmitted feature of device \(k\) as \(\Bz_k=[\Bu_k^\top,\Bs_k^\top]^\top\), where \(\Bu_k,\Bs_k\in\C^2\) represent the nuisance and the semantic sub-features, respectively. Then, the corresponding covariance matrix is partitioned as \(\BSigma^{(kk)}=\begin{bmatrix}\BSigma_u^{(kk)}&\BSigma_{us}^{(kk)}\\ \BSigma_{su}^{(kk)}&\BSigma_s^{(kk)}\end{bmatrix}\). Hence, the energy of the two sub-features is given by $\Tr \BSigma_u^{(kk)}$ and $\Tr \BSigma_s^{(kk)}$, respectively.
We note that the nuisance sub‑feature has high energy but weak standalone label information, and it is not simply artificial noise.
\par
For further insights, Table~\ref{tab:feature_probe_diagnostics} summarizes the diagnostic results across the different feature-roles. The diagnostics are evaluated on the given test set, and the five encoders are independently trained with the same architecture. It can be observed that the energy of the nuisance sub-feature is larger than that of the semantic sub-feature, with an average trace ratio of \(6.09\pm0.63\). However, the linear classification accuracy for the nuisance feature is only \(48.52\pm7.64\%\), whereas the semantic sub-feature achieves \(96.19\pm0.92\%\). Moreover, the full feature vector $[\Bu, \Bs]$ also exhibits strong linear separability, with a probe accuracy of \(96.32\pm0.83\%\). These discrepancies in probe accuracy and feature energy indicate an energy-discriminability mismatch within the full feature: The sub-feature with high energy is weakly informative when used alone, whereas lower-energy directions carry more discriminative information. The condition numbers of the covariance matrices of the two sub-features are also listed in Table \ref{tab:feature_probe_diagnostics}. Notably, the observed discrepancies are not artifacts of the ill‑conditioning of the covariance matrix. Collectively, the diagnostics establish the setting of the controlled feature that is used in the following experiments. 


\begin{figure*}[t]
\centering
\subfloat[]{%
\includegraphics[width=0.32\textwidth]{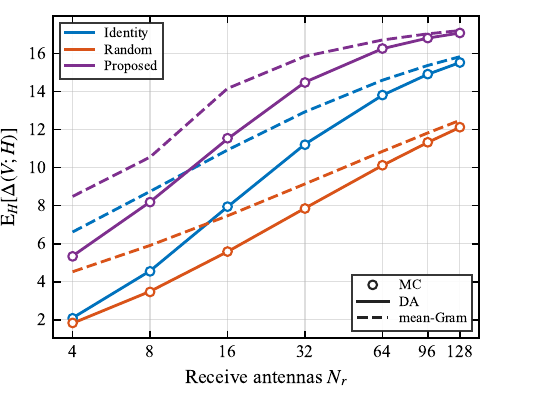}%
\label{fig:de_validation_nr_utility}%
}
\hfill
\subfloat[]{%
\includegraphics[width=0.30\textwidth]{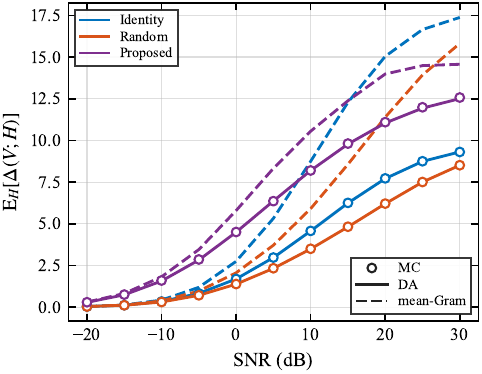}%
\label{fig:de_validation_snr_utility}%
}
\hfill
\subfloat[]{%
\includegraphics[width=0.32\textwidth]{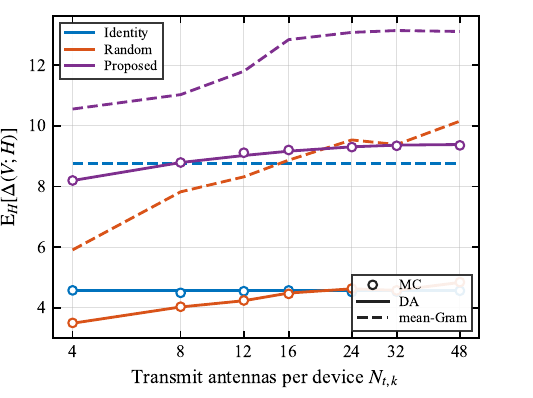}%
\label{fig:de_validation_nt_utility}%
}
\caption{Deterministic approximation validation of the task utility. The three subfigures compare Monte-Carlo estimates of \(\mathbb E_\BH[\Delta(V;\BH)]\), the finite-ratio deterministic approximation, and the mean-Gram approximation for the identity, random, and proposed precoders.}
\label{fig:de_validation_utility}
\end{figure*}

\subsection{Baselines}
\label{subsec:experiment_baselines}
For performance comparison, we consider the following precoder baselines. For baselines 3)--5), the expected objective is approximated by sample average approximation (SAA)~\cite{kleywegt2002sample} over \(M\) offline training channel samples \(\{\BH^{(m)}\}_{m=1}^{M}\).

\begin{itemize}
    \item[1)] \textbf{Identity}: Each device uses a scaled rectangular identity precoder. Specifically, let \(\BE_k\in\C^{N_{t,k}\times d_k}\) be the matrix with ones on the first \(\min(N_{t,k},d_k)\) diagonal entries and zeros elsewhere. We set \(\BV_k=a_k\BE_k\), where \(a_k\) is chosen to satisfy \eqref{eq:power_constraint}.
    \item[2)] \textbf{Random}: For each device, we draw \(\widetilde{\BV}_k\in\C^{N_{t,k}\times d_k}\) with i.i.d. \(\cC\cN(0,1)\) entries and set \(\BV_k=a_k\widetilde{\BV}_k\), where \(a_k\) is chosen to satisfy \eqref{eq:power_constraint}. The reported result is averaged over 8 independent realizations.
    \item[3)] \textbf{Maximum-ER Eigenvalue SAA~\cite{jorswieck2004capacityRange}}: Write the eigenvalue decomposition for \(\BT_k\) as \(\BT_k=\BU_k\BLambda_k\BU_k^H\). Device \(k\in[K]\) transmits along the eigenvectors of \(\BT_k\)~\cite[Theorem~1]{jorswieck2004capacityRange}, and the corresponding transmission covariance is parameterized as \(\BQ_k=\BU_k\operatorname{diag}(\Bq_k)\BU_k^H\), where $\{q_{k, i}\}_{k, i = 1}^{K, N_{t,k}}$ are obtained by maximizing the ergodic rate (ER)~\cite[Eq. (5)]{jorswieck2004capacityRange}. The resulting covariance blocks are then converted to feasible feature precoders\footnote{The solution for $\BV_k$ exists if $N_{t, k} \leq d_k$ and $\BSigma^{(kk)}$ is full rank.  In the experiment, we set $N_{t, k} = d_k = 4$ and $\BV_k = \BQ^{1/2}_k (\BSigma^{(kk)})^{-1/2}$.} \(\{\BV_k\}_{k=1}^{K}\).  
    \item[4)] \textbf{Maximum-ER Covariance SAA}: This baseline serves as a less constrained ER-oriented counterpart compared to baseline 3). Instead of parameterizing \(\BQ_k\) by the eigenvectors of \(\BT_k\), this baseline directly optimizes the transmit covariance matrices \(\{\BQ_k\succeq0\}_{k=1}^{K}\). The optimization maximizes the same SAA ER objective under \(\Tr(\BQ_k)\le N_{t,k}P_k\). The resulting covariance blocks are then converted to feasible feature precoders \(\{\BV_k\}_{k=1}^{K}\).
     \item[5)]  \textbf{LMMSE SAA~\cite{cai2024taskcomm}}: 
     For each channel sample $\BH^{(m)} \in \{\BH^{(m)}\}_{m=1}^{M}$, the edge receiver observes \(\By_m=\BH^{(m)}\BV\Bz+\Bn_m\) and applies a linear minimum mean-square error (LMMSE) detector \(R_m(\BV) \in \C^{d \times N_r} \). The precoder is obtained by solving
    \begin{align}
        \min_{\BV\in\cV}
        \frac{1}{M}\sum_{m=1}^{M}
        \E \left[ \left\|\Bz-R_m(\BV) \By_m\right\|^2 \right].
    \end{align}

    \item[6)] \textbf{Maximum-ER Water-Filling~\cite{couillet2011macde}}: 
    Using Theorem \ref{thm: EMI approx}, the deterministic approximation for the ER can be derived. The covariance matrix $\BQ = \mathrm{blkdiag}(\BQ_k; k \in [K])$ is obtained by optimizing the deterministic approximation under the power constraint via an iterative water-filling method \cite[Table~II]{couillet2011macde}. Finally, $\BV_k$ is obtained by $\BQ_k$ using the same method as in baseline 3). 

    \item[7)] \textbf{Maximum-MCR$^2$ Oracle~\cite{cai2024taskcomm}}: 
    In this baseline, the transmitter knows the instantaneous $\BH$, and the precoder is obtained by solving 
    \begin{align}
    \max_{\BV \in \cV} \Delta(\BV; \sigma^2, \BH)
    \end{align}
    using the BCA method. This precoder is used only as an oracle reference and is not counted as a statistical-CSIT baseline.
\end{itemize}

Table~\ref{tab:baseline_information_sets} summarizes all the considered methods. All reported results are evaluated with the same frozen feature encoder, channel model, power constraints, MAP receiver, and fixed ModelNet10 test split.


\subsection{Accuracy of the Theoretical Result}
\label{subsec:de_validation}

Fig.~\ref{fig:de_validation_utility} validates the accuracy of Theorem~\ref{thm: EMI approx}. For each precoder, we compare the Monte-Carlo simulation of \(\mathbb \E[\Delta(\BV;\sigma^2, \BH)]\), the deterministic approximation \eqref{eq: de MI}, and the mean-Gram approximation \(\Delta_{\rm MG}(\BV; \sigma^2)\)\footnote{For any covariance matrix \(\bcC\in \widetilde{\Sigma}\), the mean-Gram approximation is given by \(\E \log\det(\beta \BI_{N_r}+\alpha \BH\BV\bcC\BV^H\BH^H)\approx \log\det(\beta\BI_{N_r}+\alpha \bcC\BV^H\mathbb E[\BH^H\BH]\BV)\), where the right-hand side follows from the identity $\det(\BI + \BX\BY) = \det(\BI + \BY\BX)$ and the replacement \(\BH^H\BH\gets \mathbb E[\BH^H\BH]\).}. The mean-Gram curve is included as an intuitive benchmark because it preserves the log-det form of the semantic utility while replacing channel by the Gram matrix, whereas the Monte-Carlo simulation serves as an empirical ground truth. All three subfigures evaluate three precoders: identity, random, and the proposed BCA-SCA precoder.
\par
In Fig.~\ref{fig:de_validation_nr_utility}, we set $\textrm{SNR} = 10\ dB$, \(N_{t,k}=4\), and vary \(N_r\) from 4 to 128. In Fig.~\ref{fig:de_validation_snr_utility}, we set $N_r=8$, $N_{t,k}=4$, and vary the SNR from \(-20\) dB to 30 dB. In Fig.~\ref{fig:de_validation_nt_utility}, we set $N_r=8$, $\textrm{SNR} = 10\ dB$, and vary \(N_{t,k}\) from 4 to 48. Across the three settings, the proposed deterministic approximation closely matches the Monte-Carlo estimates of~\eqref{eq:statistical_csit_objective}. This agreement provides a numerical validation of Theorem~\ref{thm: EMI approx} and confirms that the deterministic approximation remains accurate in the finite-dimensional regimes. Notably, the mean-Gram approximation deviates largely from the Monte-Carlo curves, especially in the high SNR and small $N_r$ regime. This is because $\BH^H\BH$ is replaced by $\mathbb E[\BH^H\BH]$ before applying the nonlinear $\log\det(\cdot)$ operation. Such replacement removes the finite-dimensional eigenvalue fluctuations of $\BH^H\BH$, which become important in the high SNR or small $N_r$ regime. By contrast, the proposed deterministic approximation captures the finite-dimensional eigenvalue distribution through the fixed-point equations, which explains its closer agreement with the Monte-Carlo curves in Fig.~\ref{fig:de_validation_utility}.



\begin{figure*}[t]
\centering
\subfloat[]{%
\includegraphics[width=0.32\textwidth]{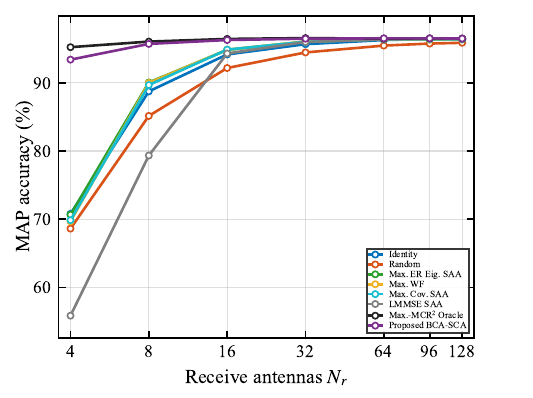}%
\label{fig:system_performance_nr}%
}
\hfill
\subfloat[]{%
\includegraphics[width=0.32\textwidth]{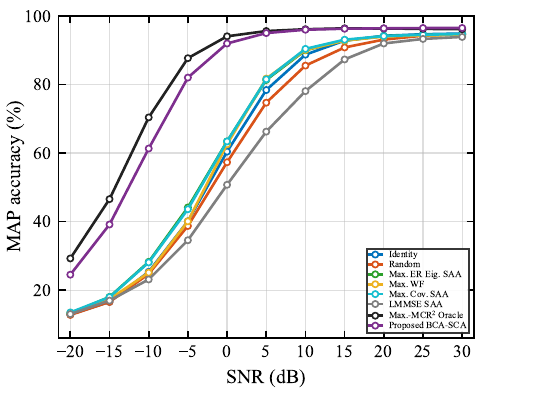}%
\label{fig:system_performance_snr}%
}
\hfill
\subfloat[]{%
\includegraphics[width=0.32\textwidth]{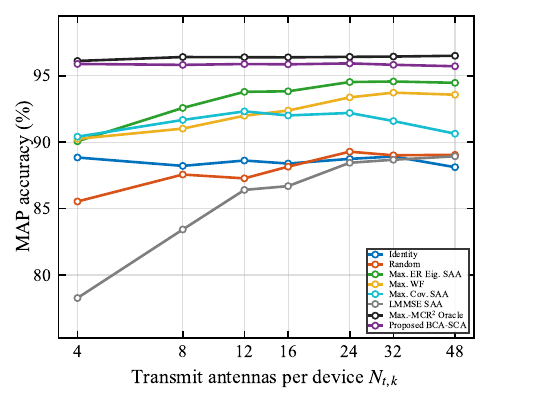}%
\label{fig:system_performance_nt}%
}
\caption{MAP classification accuracy comparison. The setting for each subfigure is as follows: (a) \(N_{t,k}=4\), SNR = 10 dB, and \(N_r \in \{4,8,16,32,64,128\}\); (b) \(N_{t,k}=4\), \(N_r=8\), and SNR \(\in \{-20,-10,0,10,20,30\}\) dB; (c) \(N_r=8\), SNR = 10 dB, and \(N_{t,k} \in \{4,8,16,32,48\}\).}
\label{fig:system_performance_comparison}
\end{figure*}
\subsection{System Performance Comparison}
\label{subsec:system_performance}

Fig.~\ref{fig:system_performance_comparison} compares the MAP classification accuracy of the proposed task-oriented precoder with the baseline methods. In Fig.~\ref{fig:system_performance_nr}, we set $\textrm{SNR} = 10\ dB,\ N_{t,k}=4 $ and vary \(N_r\) from 4 to 128. In Fig.~\ref{fig:system_performance_snr}, we set $N_r=8$, $N_{t,k}=4$, and vary the SNR from \(-20\) dB to 30 dB. In Fig.~\ref{fig:system_performance_nt}, we set $N_r=8$, $\textrm{SNR} = 10\ dB$, and vary \(N_{t,k}\) from 4 to 48.

Across all three sweeps, the proposed precoder is close to the oracle method in the moderate-to-high \(\mathrm{SNR}\) and large \(N_r\) regimes. The proposed method also outperforms all non-oracle baselines, especially in small $N_r$ and low $\textrm{SNR}$ regimes. This indicates that when the wireless channel remains a bottleneck for downstream task, transmit power allocation directly affects which feature components are reliably transmitted to the receiver. The proposed method optimizes an MCR{$^2$} objective as a task utility surrogate, such criterion encourages transmit power to be allocated towards semantic feature components rather than high energy but less informative nuisance feature components. 
By contrast, the capacity-oriented baselines and the LMMSE baseline prioritize spectral efficiency and mean-square feature reconstruction rather than the contribution of each feature component to the downstream task. Since the nuisance block has larger feature energy in our controlled feature construction, reconstruction-oriented methods tend to preserve nuisance components, while capacity-oriented methods maximize an ergodic rate objective that are agnostic to the nuisance or semantic role of the transmitted features. Hence, neither type of method explicitly aligns the resulting precoder with the class-discriminative feature directions required for MAP classification. Furthermore, as \(N_r\) or \(\mathrm{SNR}\) increases, the performance of all methods improves and the gaps among methods become smaller. This is because the wireless channel becomes less limiting for the downstream task and the accuracy of all methods eventually approaches the ceiling performance imposed by the fixed feature encoder.



\subsection{Semantic Role-Power Allocation}
\label{subsec:mode_allocation_diagnostics}

To interpret the accuracy gains in Fig.~\ref{fig:system_performance_comparison}, we examine how the transmit covariance induced by the precoder is distributed between the two feature components defined in Section~\ref{subsec:feature_role_setup}. 

For device \(k\), we write the extracted feature as \(\Bz_k=[\Bu_k^\top,\Bs_k^\top]^\top\), where \(\Bu_k\) and \(\Bs_k\) denote the nuisance and semantic blocks, respectively. We partition the precoder and the device-wise feature covariance as
\begin{align}
\label{eq: partition precoder cov}
    \BV_k = \begin{bmatrix}
        \BV_{k, u} & \BV_{k, s}
    \end{bmatrix}, ~~ \BSigma^{(kk)} = \begin{bmatrix}
        \BSigma_{u}^{(kk)} & \BSigma_{us}^{(kk)} \\
        \BSigma_{su}^{(kk)} & \BSigma_{s}^{(kk)}
    \end{bmatrix},
\end{align}
respectively. With \(N_{t,k} P_k^{\rm tx}=\operatorname{Tr}(\BV_k\BSigma^{(kk)}\BV_k^H)\), the nuisance, semantic, and \(u\)-\(s\) cross role-power fractions are defined as
\begin{align*} 
p_{k,u}&=\frac{\Tr(\BV_{k,u}\BSigma_{u}^{(kk)}\BV_{k,u}^H)}{N_{t,k} P_k^{\rm tx}},\quad
p_{k,s}=\frac{\Tr(\BV_{k,s}\BSigma_{s}^{(kk)}\BV_{k,s}^H)}{N_{t,k} P_k^{\rm tx}}, \\
p_{k,us}&=\frac{\Tr(\BV_{k,u}\BSigma_{us}^{(kk)}\BV_{k,s}^H+\BV_{k,s}(\BSigma_{us}^{(kk)})^H\BV_{k,u}^H)}{N_{t,k} P_k^{\rm tx}}. \numberthis
\end{align*}
Table~\ref{tab:role_power_10db} presents these role-power fractions at 10 dB under the default antenna setting. For all methods, we report the mean and standard deviation of each entry which are computed over the five frozen feature encoders, the three edge device and method-specific random seeds or channel realizations when applicable.

The results show a clear two-group structure. The proposed precoder and the oracle method both allocate most transmitted covariance to the semantic block, whereas the remaining methods retain substantially more covariance on the nuisance block. 
This indicates that the proposed precoder based on statistical-CSIT allocates transmitted covariance according to task relevance, rather than simply following high-energy feature directions. 
This role allocation provides a mechanism-level explanation for the accuracy gain in Fig.~\ref{fig:system_performance_comparison}. The \(u\)-\(s\) cross entry can be negative because it is an interaction term in the decomposition of \(\operatorname{Tr}(\BV_k\BSigma^{(kk)}\BV_k^H)\), not an independent nonnegative power component.

\begin{table*}[t]
\centering
\caption{Semantic role-power allocation at 10 dB.}
\label{tab:role_power_10db}
\scriptsize
\begin{tabular}{lrrrr}
\toprule
Method & Nuisance & Semantic & \(u\)-\(s\) cross & MAP accuracy \\
\midrule
Proposed BCA-SCA & \(0.067\pm0.072\) & \(0.955\pm0.070\) & \(-0.023\pm0.021\) & \(95.97\pm0.40\) \\
Identity & \(0.856\pm0.019\) & \(0.144\pm0.019\) & \(0.000\pm0.000\) & \(88.70\pm2.23\) \\
Random & \(0.851\pm0.069\) & \(0.154\pm0.062\) & \(-0.005\pm0.035\) & \(85.51\pm2.03\) \\
Maximum-ER Eigenvalue SAA & \(0.746\pm0.198\) & \(0.282\pm0.199\) & \(-0.028\pm0.023\) & \(89.97\pm1.45\) \\
Maximum-ER Covariance SAA & \(0.765\pm0.184\) & \(0.263\pm0.186\) & \(-0.027\pm0.023\) & \(90.38\pm1.15\) \\
LMMSE SAA & \(0.954\pm0.046\) & \(0.031\pm0.049\) & \(0.015\pm0.011\) & \(78.07\pm4.59\) \\
Maximum-ER Water Filling & \(0.745\pm0.200\) & \(0.283\pm0.202\) & \(-0.028\pm0.024\) & \(90.06\pm1.59\) \\
Maximum-MCR{$^2$} Oracle & \(0.157\pm0.100\) & \(0.865\pm0.093\) & \(-0.023\pm0.023\) & \(96.12\pm0.36\) \\
\bottomrule
\end{tabular}
\end{table*}

We next examine whether the same power allocation is reflected at the eigenmode level. Specifically, we decompose the precoded transmit covariance and the corresponding receive-side signal covariance as follows
\begin{align*}
 \BP_{\BSigma}(\BV) &= \BV\BSigma \BV^H = \BP_u+\BP_s+\BP_{us, su} \\
 &=\sum_i\lambda_i \Ba_i \Ba_i^H, \numberthis\\
\BS_{\BSigma}(\BV, \BH) &= \rho_{\rm mcr}\BH\BV\BSigma \BV^H\BH^H\\
&=\BS_u+\BS_s+\BS_{us, su}=\sum_i\mu_i \Bb_i \Bb_i^H, \numberthis
\end{align*}
where \(\BP_u\), \(\BP_s\), and \(\BP_{us, su}\) are induced by the partition \eqref{eq: partition precoder cov}, and \(\BS_u=\rho_{\rm mcr}\BH\BP_u\BH^H\), \(\BS_s=\rho_{\rm mcr}\BH\BP_s\BH^H\), and \(\BS_{us, su}=\rho_{\rm mcr}\BH\BP_{us, su}\BH^H\). 
Fig.~\ref{fig:role_colored_eigenspectrum} visualizes the above eigenmode decompositions. Fig.~\ref{fig:role_colored_transmit_eigenspectrum} shows the sorted normalized eigenvalues \(\lambda_i/ \Tr (\BP_{\BSigma})\) of the precoded transmit signal covariance, while Fig.~\ref{fig:role_colored_receive_eigenspectrum} shows the sorted normalized eigenvalues \(\mu_i/\Tr (\BS_{\BSigma})\) of the receive signal covariance. The bubble area is proportional to the corresponding normalized eigenvalue, while the bubble color represents the following semantic fraction
\begin{align}
\phi_i^{\rm tx}=\frac{\Ba_i^H\BP_s\Ba_i}{\Ba_i^H(\BP_u+\BP_s)\Ba_i},\qquad
\phi_i^{\rm rx}=\frac{\Bb_i^H\BS_s\Bb_i}{\Bb_i^H(\BS_u+\BS_s)\Bb_i}. 
\end{align}
Here, \(\phi_i^{\rm tx}\) and \(\phi_i^{\rm rx}\) quantify the fraction of the \(i\)-th transmit-side or receive-side eigenmode that is contributed by the semantic block. A larger \(\phi_i\) means that the corresponding eigenmode is mainly induced by the semantic block.


Fig.~\ref{fig:role_colored_eigenspectrum} shows that the two-group structure in Table~\ref{tab:role_power_10db} is also reflected at the eigenmode level. For the proposed statistical-CSIT precoder, the dominant transmit-side and receive-side eigenmodes have large semantic fractions, and the oracle method exhibits a similar pattern. By contrast, the dominant eigenmodes of the remaining baselines have much smaller semantic fractions, indicating that their strongest modes are still mainly induced by the nuisance block. Therefore, the eigenmode-level diagnostic supports the same conclusion as Table~\ref{tab:role_power_10db}: the strongest transmit-side and receive-side eigenmodes of the proposed statistical-CSIT precoder are mainly induced by the semantic block, whereas those of the baseline methods are mainly induced by the nuisance block.


\begin{figure*}[t]
\centering
\subfloat[]{%
\includegraphics[width=0.49\textwidth]{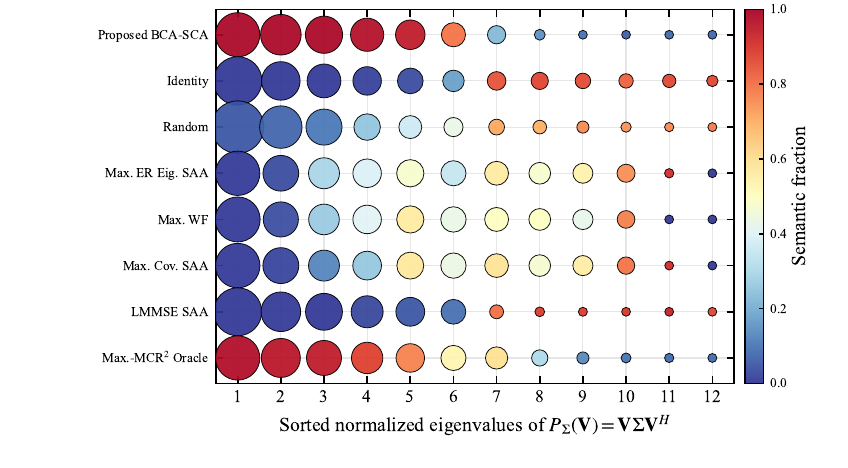}%
\label{fig:role_colored_transmit_eigenspectrum}%
}
\subfloat[]{%
\includegraphics[width=0.46\textwidth]{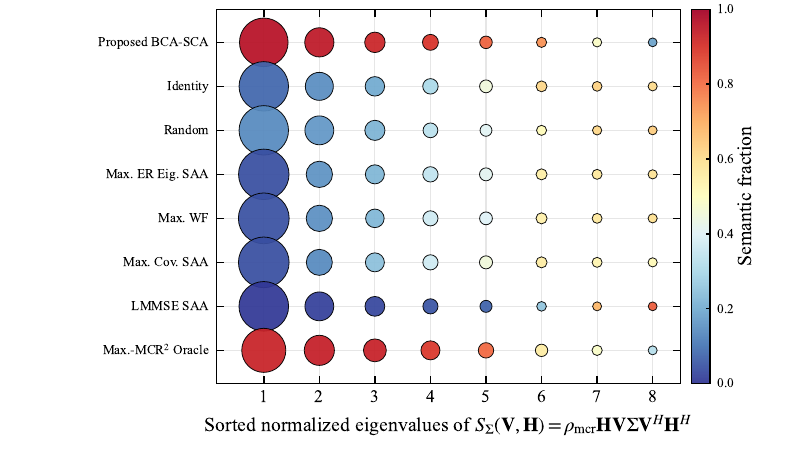}%
\label{fig:role_colored_receive_eigenspectrum}%
}
\caption{Semantic composition of sorted normalized eigenvalues. The transmit-side plot uses \(\BP_{\BSigma}(\BV)=\BV\BSigma \BV^H\), while the receive-side plot uses \(\BS_{\BSigma}(\BV,\BH)=\rho_{\rm mcr}\BH\BV\BSigma \BV^H\BH^H\). Bubble area represents normalized eigenvalue mass, and color represents the semantic fraction of the corresponding mode.}
\label{fig:role_colored_eigenspectrum}
\end{figure*}

\section{Conclusion}
\label{sec:conclusion}

In this paper, we studied the task-oriented precoding problem for multi-device MIMO systems with statistical CSIT. We adopted the received-feature MCR\(^2\) utility as a surrogate for inference accuracy, and formulated it as a coupled difference of log-det terms. By leveraging RMT, we derived deterministic approximations for the log-det terms in the surrogate objective. The resulting approximation depends only on channel statistics and training-set feature statistics, enabling precoder optimization without instantaneous CSIT. Based on this deterministic approximation objective, we developed a BCA-SCA algorithm under per-device power constraints. Numerical experiments on the ModelNet10 dataset verify the accuracy of the proposed deterministic approximation and demonstrate the superior system performance of the statistical-CSIT precoder among all baselines. Future work may extend the framework to non-centered channel models and second-order analyses for reliability-constrained statistical-CSIT precoding.

\appendices

\section{Proof of Theorem \ref{thm: convergence resolvent}}
\label{app: convergence resolvent}
Before going to the details, we define the key intermediate quantities below:
\begin{align*}
    &g_k(z) = \frac{1}{M_k} \Tr \BOmega_k \BQ(z), \quad \underline{g}_k(z) = \E g_k(z), \\
    &\widetilde{\BPsi}(z) = -\frac{1}{z}\left[\BI_{N_t} + \sum_{k \in [K]} \underline{g}_k(z) \widetilde{\BOmega}_k  \right]^{-1}, \\
    &\widetilde{g}_k(z) = \frac{1}{M_k} \Tr \widetilde{\BOmega}_k \widetilde{\BPsi}(z), \\
    & \BPsi(z) = -\frac{1}{z} \left[ \BI_{N_r} + \sum_{k \in [K]} \widetilde{g}_k(z) \BOmega_k \right]^{-1}. \numberthis \label{eq: resol intermediat terms}
\end{align*}
Our proof proceeds in two steps: a stochastic step and a deterministic step. In the stochastic step, we first show that $\BQ(z)$ can be approximated by $\BPsi(z)$, i.e., the two matrices are close in the weak operator sense. In the deterministic step, we use the bound on $\BQ - \BPsi$  obtained in the stochastic step to establish the connection between $\BPsi$ and $\BTheta$.
\par
The stochastic step is completed using Gaussian tools \cite{hachem2008mimo}, which consist of two main ingredients: the integration by parts formula and the Nash–Poincaré inequality. First, employing the resolvent identity 
\begin{align}
\label{eq: resol identity}
    \BQ(z) \BZ\BZ^H - z\BQ(z) = \BI_{N_r},
\end{align}
we apply the integration by parts formula to decompose $[\BQ(z) \BZ\BZ^H]_{i, j}$ into products of resolvent functional terms (see \eqref{eq: resol intermediat terms} below). Then, the Nash–Poincaré inequality is used to control the variance for the related terms.

Expanding the matrix product, we have
\begin{align*}
&\E [\BQ(z) \BZ \BZ^H]_{i, j} \\
&=
    \sum_{t, m, s', s, k}\E [\BQ(z) \BA_k]_{i, t} [\BX_k]_{t, m} [\widetilde{\BA}_k]_{m, s'} [\BZ]_{j, s}^*. \numberthis
\end{align*}
Applying the integration by parts formula to each term of RHS of the above equation and taking $[\BX_k]_{t, m}$ as the variable, we obtain
\begin{align*}
    &\E [\BQ \BA_k]_{i, t} [\BX_k]_{t, m} [\widetilde{\BA}_k]_{m, s'} [\BZ]_{ j, s}^* \\
    &= \frac{[\widetilde{\BA}_k]_{m, s'} }{M_k} \E \frac{\partial  [\BQ\BA_k]_{i, t} [\BZ]_{j,s}^*}{\partial [\BX_k]_{t, m}^*} \\
    &= -\frac{[\widetilde{\BA}_k]_{m, s'} }{M_k} \E \left[ [{\BQ}\BZ\widetilde{\BA}_k^H]_{i, m} [\BA_k^H\BQ \BA_k]_{t,t} [\BZ]^*_{j,s} \right] \\
    &+ \sum_k \frac{[\widetilde{\BA}_k]_{m, s'}[\widetilde{\BA}_k]_{m, s}^* [\BA_k^H]_{t, j} }{M_k} \E \left[ [\BQ\BA_k ]_{i, t} \right]. \numberthis
\end{align*}
Summing over the subscripts $k$, $t$, and $m$, we obtain
\begin{align*}
    &\E [\BQ \BZ]_{i, s'}  [\BZ]_{ j, s}^* \\
    &=  \E \left[ -\sum_kg_k[{\BQ}\BZ   \widetilde{\BOmega}_k]_{i, s'}  [\BZ]^*_{j,s} \right] \\
    &+ \frac{[\widetilde{\BOmega}_k]_{s, s'}}{M_k} \E [\BQ\BOmega_k]_{i, j}. \numberthis
\end{align*}
Writing the random variable $g_k = \E g_k + g_k^\circ$, multiplying $-z[\widetilde{\BPsi}]_{s', s}$ on both sides of the above equation, and summing over the subscripts $s$ and $s'$ yields
\begin{align*}
    \E [\BQ \BZ \BZ^H]_{i, j} &= - \sum_k z\widetilde{g}_k\E [\BQ \BOmega_k]_{i, j} \\
    &+z \sum_k \mathrm{Cov}(g_k, [{\BQ}\BZ \widetilde{\BOmega}_k \widetilde{\BPsi}\BZ^H]_{i, j} ). \numberthis
\end{align*}
Applying the resolvent identity \eqref{eq: resol identity}, multiplying $[\BPsi \bcR]_{j, i}$ at both sides of the above equation, and summing over $i$ and $j$, we obtain 
\begin{align}
    \Tr \bcR \E \BQ - \Tr \bcR\BPsi = \sum_k \varepsilon_{\bcR, k}, 
    \label{eq: Tr Q Psi}
\end{align}
where $ \varepsilon_{\bcR, k} = -z\Cov(g_k, \Tr{\BQ}\BZ \widetilde{\BOmega}_k \widetilde{\BPsi} \BZ^H \BPsi \bcR)$. To derive the bound for $\varepsilon_{\cR, k}$, it suffices to control the variances of $g_k$ and $h_k = \Tr{\BQ}\BZ \widetilde{\BOmega}_k \widetilde{\BPsi} \BZ^H \BPsi \bcR$, respectively and then apply the Cauchy-Schwarz inequality $\E^{1/2} |XY| \leq \E |X|^2 \E |Y|^2$. This can be achieved using the Nash-Poincaré inequality. We only state the results below and refer the reader to \cite{hachem2008mimo} for the detailed derivations. In particular, we have 
\begin{align*}
    \Var(g_k) = \cO\left(N_r^{-2} |z|^{-4}\right), ~~ \Var(h_k) = \cO\left(|z|^{-6} + |z|^{-8}\right), \numberthis
\end{align*}
for $z < 0$.  As a result, from \eqref{eq: Tr Q Psi}, we derived 
\begin{align}
\label{eq: trace bound intermidiate}
    \Tr \bcR (\E \BQ - \BPsi) = \cO\left(N_r^{-1}(|z|^{-4} + |z|^{-5})\right),
\end{align}
which establishes the convergence of $\BQ - \BPsi$ in the weak operator sense. 
\par
For the deterministic step, we connect $\BPsi$ with $\BTheta$ (recall the definition in \eqref{eq: de}). This can be accomplished by examining the differences $\underline{g}_k - \delta_k$ and the derivatives of $\underline{g}_k $ and $\delta_k$, respectively. Define $\Bdelta = [\delta_1, \ldots, \delta_K]^\top$ and $\underline{\Bg} = [\underline{g}_1, \ldots, \underline{g}_K]^\top$. By the resolvent identity $\BX - \BY = \BX(\BY^{-1} - \BX^{-1}) \BY$ and \eqref{eq: de}, we have
\begin{subequations}
\begin{align}
&(\BI_K - z^2\BGamma_{\delta} \widetilde{\BGamma}_{\delta})\Bdelta' =  \Bv_\delta, \label{eq: diffdelta} \\
&(\BI_K - z^2\BGamma_{g} \widetilde{\BGamma}_{g}) \underline{\Bg}' = \Bv_{g} + \mathbf{1}_K \cO(N_r^{-2}q_z), \label{eq: diff g} \\
&(\BI_K - z^2\BGamma_{g, \delta} \widetilde{\BGamma}_{g, \delta})(\underline{\Bg} - \Bdelta) =  \mathbf{1}_K \cO(N_r^{-2}q_z),  \label{eq: diff delta g}
\end{align}
\end{subequations}
where $q_z = |z|^{-4} + |z|^{-5}$ and
\begin{align*}
& [\Bv_{\delta}]_k = \frac{\Tr \BOmega_k \BTheta^2}{M_k} , ~~ [\Bv_{g}]_k = \frac{\Tr \BOmega_k \BPsi^2}{M_k} , \\ 
&[\BGamma_{\delta}]_{k, l} = \frac{\Tr \BOmega_k \BTheta \BOmega_l\BTheta}{M_k},  [\widetilde{\BGamma}_{\delta}]_{k, l} = \frac{\Tr \widetilde{\BOmega}_k \widetilde{\BTheta} \widetilde{\BOmega}_l \widetilde{\BTheta}}{M_k},\\
&[\BGamma_{g}]_{k, l} = \frac{\Tr \BOmega_k \BPsi \BOmega_k \BPsi}{M_k}, [\widetilde{\BGamma}_{g}]_{k, l} = \frac{\Tr \widetilde{\BOmega}_k\widetilde{\BPsi} \widetilde{\BOmega}_l \widetilde{\BPsi}}{M_k} ,\\
&[\BGamma_{g, \delta}]_{k, l} =  \frac{\Tr \BOmega_k \BPsi  \BOmega_l \BTheta}{M_k}, [\widetilde{\BGamma}_{g, \delta}]_{k, l} = \frac{\Tr \widetilde{\BOmega}_k \widetilde{\BPsi} \widetilde{\BOmega}_l \widetilde{\BTheta}}{M_k}. \numberthis 
\end{align*}
Here, \eqref{eq: diffdelta} and \eqref{eq: diff g} are derived to solve \eqref{eq: diff delta g} and evaluate the upper bounds for $\underline{\Bg} - \Bdelta$. In particular, one can verify that there exists a constant $C > 0$ such that  
\begin{equation}
\frac{C^{-1}}{(|z| + C)^2} \leq x \leq \frac{C}{|z|^2}    
\end{equation}
for any symbol $x $  belongs to $ \{[\Bv_\delta]_k, [\Bv_g]_k, [\Bdelta']_k, [\underline{\Bg}']_k\}_{k = 1}^K$. As a result, the matrix $\BI_K - z^2\BGamma_g \widetilde{\BGamma}_g$ is invertible when $N_r^{-2}q_z \leq C'/(|z| + C)^2$ for some constant $C'$ depending on $C$. By the Cauchy-Schwarz inequality $|\Tr \BX\BY^H| \leq (\Tr \BX\BX^H \Tr \BY\BY^H)^{1/2}$, we know that $|[\BGamma_{g, \delta}\widetilde{\BGamma}_{g, \delta}]_{k, l}| \leq ([\BGamma_{\delta}\widetilde{\BGamma}_{\delta}]_{k, l}[\BGamma_{g}\widetilde{\BGamma}_{g}]_{k, l})^{1/2}$. 
Using \cite[Lemmas 7, 8]{zhuang2025no}, it follows that $\BI_K - z^2\BGamma_{g, \delta}\widetilde{\BGamma}_{g, \delta}$ is also invertible in the region $\{z: N_r^{-2} q_z \leq C'/(|z| + C)^2\}$. A further analysis similar to that in \cite{zhuang2025no} yields
\begin{align}
    \norm{\underline{\Bg} - \Bdelta}_{\infty} = \cO\left( \frac{\mathsf{P}_4(|z|^{-1})}{N_r^2|z|^3} \right), ~~ z < 0.
\end{align}
where $\mathsf{P}_4$ is a polynomial of degree at most 4 with positive coefficients. Applying the resolvent identity yields 
\begin{align}
    \Tr \bcR \BPsi - \Tr \bcR \BTheta = \cO \left( \frac{P_4(|z|^{-1})}{N_r|z|^5} \right).
\end{align}
Combining the above bound with $\eqref{eq: trace bound intermidiate}$, we can get
\begin{align}
    \Tr \bcR (\E \BQ - \BTheta) = \cO\left( \frac{\cP_5(|z|^{-1}) }{N_r |z|^4}\right).
\end{align}
Therefore, we have completed the proof for Theorem \ref{thm: convergence resolvent}. \QED

\section{Proof of Theorem \ref{thm: EMI approx}}
\label{app: proof of thm EMI approx}
In this section, we apply Theorem \ref{thm: convergence resolvent} to prove Theorem \ref{thm: EMI approx}. 
We define the resolvents $\BQ_{\bcC}(x) = (-x\BI_{N_r} + \BH \BP_{\bcC} \BH^H)$, for $\bcC \in \widetilde{\Sigma}$. Setting the parameters $z = -\sigma^2$, $L_k = N_{r}$, $M_k = N_{t, k}$,  $\BA_k = \BB_k^{1/2}$, and $\widetilde{\BA}_k = \BPi_k\bcT_k^{1/2}\BP_{\bcC}^{1/2}$, where 
\begin{align}
    \BPi_k = \begin{bmatrix}
        \mathbf{0}_{N_{t, k} \times N_{t, 1}} &         \mathbf{0}_{N_{t, k} \times N_{t, 2}} & \ldots& \BI_{N_{t, k}}&\ldots
    \end{bmatrix} 
\end{align}
denotes the selection matrix of size $N_{t, k} \times N_t$, we obtain from Theorem \ref{thm: convergence resolvent} the fixed point equation in \eqref{eq: de MCR2 equation}. Moreover, the resolvent $ \BQ_{\bcC}(-x)$ can be approximated by $\BTheta_{\bcC}(-x)$, where 
\begin{align}
    \BTheta_{\bcC}(-x) = \frac{1}{x}\Bigg[ \BI_{N_r} + \sum_{k \in [K]} \widetilde{\delta}_{\cC, k}(-x) \BB_k \Bigg]^{-1}.
\end{align}
Recall the definition of $I_{\bcC}$ in \eqref{eq: I bcC term}.  Taking the derivative $\partial I_{\bcC}(\BV; x, \BH) / \partial x$ and noting that $\lim_{x \to + \infty} I_{\bcC}(\BV; x, \BH) = 0$, we have 
\begin{align*}
 I_{\bcC}(\BV; \sigma^2, \BH) &= \int_{\sigma^2}^{\infty} -\frac{\partial I_{\bcC}(\BV; x, \BH)}{\partial x} dx \\
 & = \int_{\sigma^2}^{\infty} \left( \frac{N_r}{x} - \Tr \BQ_{\bcC}(-x)  \right) dx.  \numberthis
\end{align*}
Setting $\bcR = \BI_{N_r}$ in Theorem \ref{thm: convergence resolvent}, we obtain $\Tr \E \BQ_{\bcC}(-x) = \Tr \BTheta_{\bcC}(-x) + \cO(N_r^{-1} x^{-4} {\mathsf{P}}_5(x^{-1}))$. Consequently, by Fubini's theorem,
\begin{align*}
   &\E I_{\bcC}(\BV; \sigma^2, \BH) = \int_{\sigma^2}^{\infty}\left( \frac{N_r}{x} -  \Tr \BTheta_{\bcC}(-x) \right) dx \\
   &+ \int_{\sigma^2}^{\infty} \cO\left( \frac{P_5(x^{-1})}{N_rx^4}\right) dx \\
   &= \int_{\sigma^2}^{\infty}\left( \frac{N_t}{x} -  \Tr \BTheta_{\bcC}(-x) \right) dx  + \cO(N_r^{-1}). \numberthis \label{eq: approx emi to Theta}
\end{align*}
Write $\overline{I}_{\bcC}(\BV; x) = f_{\bcC}(\{\delta_{\bcC, k}\}_{k=1}^K, \{\widetilde{\delta}_{\bcC, k}\}_{k=1}^K, x)$. Then, the derivative $\partial \overline{I}_{\bcC}(\BV; x) / \partial x$ can be given by
\begin{align*}
    \frac{\partial \overline{I}_{\bcC}(\BV; x)}{\partial x} &= \sum_{k \in [K]} \frac{\partial f_{\bcC}}{\partial \delta_{\bcC, k}} \frac{\partial \delta_{\bcC, k}}{\partial x} \\
    &+ \sum_{k \in [K]} \frac{\partial f_{\bcC}}{\partial \widetilde{\delta}_{\bcC, k}} \frac{\partial \widetilde{\delta}_{\bcC, k}}{\partial x}  + \frac{\partial f_{\bcC}}{\partial x}. \numberthis
\end{align*}
From the fixed point equations \eqref{eq: de MCR2 equation}, we obtain for each $l\in [K]$
\begin{align}
    \frac{\partial f_{\bcC}}{\partial \widetilde{\delta}_{\bcC, l}} = \Tr\BB_l \left( \BI_{N_r} + \sum_{k \in [K]} \widetilde{\delta}_{\bcC, k} \BB_k \right)^{-1} - x \delta_{\bcC, l} = 0,
\end{align}
 and similarly $\frac{\partial f_{\bcC}}{\partial {\delta}_{\bcC, l}} = 0$. Therefore, 
\begin{align*}
    &- \frac{\partial \overline{I}_{\bcC}}{\partial x} = \sum_{k \in [K]} N_{t, k} \delta_{\bcC, k}(-x) \widetilde{\delta}_{\bcC, k}(-x) \\
    & = \Tr (x\BI_{N_r})^{-1} (\sum_{k \in [K]} x\widetilde{\delta}_{\bcC, k}(-x) \BB_k)\BTheta_{\bcC}(-x) \\
    & = \frac{N_r}{x} -  \Tr \BTheta_{\bcC}(-x). \numberthis \label{eq: deri de I bcC}
\end{align*}
Using $\lim_{x \to +\infty} \overline{I}_{\bcC}(\BV; x) = 0$ and combining \eqref{eq: deri de I bcC} and \eqref{eq: approx emi to Theta}, we obtain
\begin{align}
    \E I_{\bcC}(\BV; \sigma^2, \BH) = \overline{I}_{\bcC}(\BV; \sigma^2) + \cO(N_r^{-1}),
\end{align}
which completes the proof of Theorem \ref{thm: EMI approx}. \QED

\begingroup
\small
\renewcommand{\baselinestretch}{0.94}\selectfont
\bibliographystyle{bibtex/IEEEtran}
\bibliography{rmt_mcr_references}

@article{cai2024taskcomm,
  title={Multi-device task-oriented communication via maximal coding rate reduction},
  author={Cai, Chang and Yuan, Xiaojun and Zhang, Ying-Jun Angela},
  journal={IEEE Transactions on Wireless Communications},
  volume={23},
  number={12},
  pages={18096--18110},
  year={2024},
  publisher={IEEE}
}

@article{marzetta2010noncooperative,
  title={Noncooperative cellular wireless with unlimited numbers of base station antennas},
  author={Marzetta, Thomas L},
  journal={IEEE transactions on wireless communications},
  volume={9},
  number={11},
  pages={3590--3600},
  year={2010},
  publisher={IEEE}
}

@article{love2008overview,
  title={An overview of limited feedback in wireless communication systems},
  author={Love, David J and Heath, Robert W and Lau, Vincent KN and Gesbert, David and Rao, Bhaskar D and Andrews, Matthew},
  journal={IEEE Journal on selected areas in Communications},
  volume={26},
  number={8},
  pages={1341--1365},
  year={2008},
  publisher={IEEE}
}

@book{tulino2004random,
  title={Random matrix theory and wireless communications},
  author={Tulino, Antonia M and Verd{\'u}, Sergio},
  year={2004},
  publisher={Now Publishers Inc}
}

@article{letaief2019roadmap,
  title={The roadmap to {6G}: {AI} empowered wireless networks},
  author={Letaief, Khaled B and Chen, Wei and Shi, Yuanming and Zhang, Jun and Zhang, Ying-Jun Angela},
  journal={IEEE Communications Magazine},
  volume={57},
  number={8},
  pages={84--90},
  year={2019},
  publisher={IEEE}
}

@article{letaief2022edgeai,
  title={Edge artificial intelligence for {6G}: Vision, enabling technologies, and applications},
  author={Letaief, Khaled B and Shi, Yuanming and Lu, Jianmin and Lu, Jianhua},
  journal={IEEE Journal on Selected Areas in Communications},
  volume={40},
  number={1},
  pages={5--36},
  year={2021},
  publisher={IEEE}
}

@article{reynolds2009gaussian,
  title={Gaussian mixture models.},
  author={Reynolds, Douglas A and others},
  journal={Encyclopedia of biometrics},
  volume={741},
  number={659-663},
  pages={3},
  year={2009},
  publisher={Springer City}
}

@article{kleywegt2002sample,
  title={The sample average approximation method for stochastic discrete optimization},
  author={Kleywegt, Anton J and Shapiro, Alexander and Homem-de-Mello, Tito},
  journal={SIAM Journal on optimization},
  volume={12},
  number={2},
  pages={479--502},
  year={2002},
  publisher={SIAM}
}

@inproceedings{wu20153d,
  title={3d shapenets: A deep representation for volumetric shapes},
  author={Wu, Zhirong and Song, Shuran and Khosla, Aditya and Yu, Fisher and Zhang, Linguang and Tang, Xiaoou and Xiao, Jianxiong},
  booktitle={Proceedings of the IEEE conference on computer vision and pattern recognition},
  pages={1912--1920},
  year={2015}
}

@article{gunduz2023beyond,
  title={Beyond transmitting bits: Context, semantics, and task-oriented communications},
  author={G{\"u}nd{\"u}z, Deniz and Qin, Zhijin and Aguerri, Inaki Estella and Dhillon, Harpreet S and Yang, Zhaohui and Yener, Aylin and Wong, Kai Kit and Chae, Chan-Byoung},
  journal={IEEE Journal on Selected Areas in Communications},
  volume={41},
  number={1},
  pages={5--41},
  year={2022},
  publisher={IEEE}
}

@article{wen2023task,
  title={Task-oriented over-the-air computation for multi-device edge AI},
  author={Wen, Dingzhu and Jiao, Xiang and Liu, Peixi and Zhu, Guangxu and Shi, Yuanming and Huang, Kaibin},
  journal={IEEE Transactions on Wireless Communications},
  volume={23},
  number={3},
  pages={2039--2053},
  year={2023},
  publisher={IEEE}
}

@article{Simonyan15,
  title={Very deep convolutional networks for large-scale image recognition},
  author={Simonyan, Karen and Zisserman, Andrew},
  journal={arXiv preprint arXiv:1409.1556},
  year={2014}
}

@article{cai2025end,
  title={End-to-end learning for task-oriented semantic communications over MIMO channels: An information-theoretic framework},
  author={Cai, Chang and Yuan, Xiaojun and Zhang, Ying-Jun Angela},
  journal={IEEE Journal on Selected Areas in Communications},
  volume={43},
  number={4},
  pages={1292--1307},
  year={2025},
  publisher={IEEE}
}

@article{shao2022learning,
  title={Learning task-oriented communication for edge inference: An information bottleneck approach},
  author={Shao, Jiawei and Mao, Yuyi and Zhang, Jun},
  journal={IEEE Journal on Selected Areas in Communications},
  volume={40},
  number={1},
  pages={197--211},
  year={2021},
  publisher={IEEE}
}

@article{shao2023multidevice,
  title={Task-oriented communication for multidevice cooperative edge inference},
  author={Shao, Jiawei and Mao, Yuyi and Zhang, Jun},
  journal={IEEE Transactions on Wireless Communications},
  volume={22},
  number={1},
  pages={73--87},
  year={2022},
  publisher={IEEE}
}

@article{larsson2014massive,
  title={Massive {MIMO} for next generation wireless systems},
  author={Larsson, Erik G and Edfors, Ove and Tufvesson, Fredrik and Marzetta, Thomas L},
  journal={IEEE communications magazine},
  volume={52},
  number={2},
  pages={186--195},
  year={2014},
  publisher={IEEE}
}

@article{yu2020mcr2,
  title={Learning diverse and discriminative representations via the principle of maximal coding rate reduction},
  author={Yu, Yaodong and Chan, Kwan Ho Ryan and You, Chong and Song, Chaobing and Ma, Yi},
  journal={Advances in neural information processing systems},
  volume={33},
  pages={9422--9434},
  year={2020}
}

@article{wang2024globalmcr2,
  title={A global geometric analysis of maximal coding rate reduction},
  author={Wang, Peng and Liu, Huikang and Pai, Druv and Yu, Yaodong and Zhu, Zhihui and Qu, Qing and Ma, Yi},
  journal={arXiv preprint arXiv:2406.01909},
  year={2024}
}

@article{hachem2007deterministic,
  title={Deterministic equivalents for certain functionals of large random matrices},
  author={Hachem, Walid and Loubaton, Philippe and Najim, Jamal},
  year={2007}
}

@article{hachem2008mimo,
  title={A new approach for mutual information analysis of large dimensional multi-antenna channels},
  author={Hachem, Walid and Khorunzhiy, Oleksiy and Loubaton, Philippe and Najim, Jamal and Pastur, Leonid},
  journal={IEEE Transactions on Information Theory},
  volume={54},
  number={9},
  pages={3987--4004},
  year={2008},
  publisher={IEEE}
}

@article{couillet2011macde,
  title={A deterministic equivalent for the analysis of correlated MIMO multiple access channels},
  author={Couillet, Romain and Debbah, M{\'e}rouane and Silverstein, Jack W},
  journal={IEEE Transactions on Information Theory},
  volume={57},
  number={6},
  pages={3493--3514},
  year={2011},
  publisher={IEEE}
}

@book{couillet2011book,
  title={Random matrix methods for wireless communications},
  author={Couillet, Romain and Debbah, Merouane},
  year={2011},
  publisher={Cambridge University Press}
}

@article{zhang2023doubleScattering,
  title={Asymptotic mutual information analysis for double-scattering MIMO channels: A new approach by Gaussian tools},
  author={Zhang, Xin and Song, Shenghui},
  journal={IEEE Transactions on Information Theory},
  volume={69},
  number={9},
  pages={5497--5527},
  year={2023},
  publisher={IEEE}
}

@article{zhang2024irsWiretap,
  title={Secrecy analysis for IRS-aided wiretap MIMO communications: Fundamental limits and system design},
  author={Zhang, Xin and Song, Shenghui},
  journal={IEEE Transactions on Information Theory},
  volume={70},
  number={6},
  pages={4140--4159},
  year={2023},
  publisher={IEEE}
}

@article{zhang2024nonCenteredNonSeparable,
  title={Fundamental limits of non-centered non-separable channels and their application in holographic MIMO communications},
  author={Zhang, Xin and Song, Shenghui and Letaief, Khaled B},
  journal={IEEE Transactions on Information Theory},
  year={2025},
  publisher={IEEE}
}

@article{jafar2005covariance,
  title={Multiple-antenna capacity in correlated Rayleigh fading with channel covariance information},
  author={Jafar, Syed Ali and Goldsmith, Andrea},
  journal={IEEE Transactions on Wireless Communications},
  volume={4},
  number={3},
  pages={990--997},
  year={2005},
  publisher={IEEE}
}

@article{jorswieck2004capacityRange,
  title={Channel capacity and capacity-range of beamforming in MIMO wireless systems under correlated fading with covariance feedback},
  author={Jorswieck, Eduard A and Boche, Holger},
  journal={IEEE Transactions on Wireless Communications},
  volume={3},
  number={5},
  pages={1543--1553},
  year={2004},
  publisher={IEEE}
}

@article{jorswieck2004correlation,
  title={Optimal transmission strategies and impact of correlation in multiantenna systems with different types of channel state information},
  author={Jorswieck, Eduard A and Boche, Holger},
  journal={IEEE Transactions on Signal Processing},
  volume={52},
  number={12},
  pages={3440--3453},
  year={2004},
  publisher={IEEE}
}

@article{xie2021deepsc,
  title={Deep learning enabled semantic communication systems},
  author={Xie, Huiqiang and Qin, Zhijin and Li, Geoffrey Ye and Juang, Biing-Hwang},
  journal={IEEE Transactions on Signal Processing},
  volume={69},
  pages={2663--2675},
  year={2021},
  publisher={IEEE}
}

@article{strinati2021goal,
  title={{6G} networks: Beyond Shannon towards semantic and goal-oriented communications},
  author={Strinati, Emilio Calvanese and Barbarossa, Sergio},
  journal={Computer Networks},
  volume={190},
  pages={107930},
  year={2021},
  publisher={Elsevier}
}

@article{telatar1999capacity,
  title={Capacity of multi-antenna Gaussian channels},
  author={Telatar, Emre},
  journal={European transactions on telecommunications},
  volume={10},
  number={6},
  pages={585--595},
  year={1999},
  publisher={Wiley Online Library}
}

@article{goldsmith2003capacity,
  title={Capacity limits of MIMO channels},
  author={Goldsmith, Andrea and Jafar, Syed Ali and Jindal, Nihar and Vishwanath, Sriram},
  journal={IEEE Journal on selected areas in Communications},
  volume={21},
  number={5},
  pages={684--702},
  year={2003},
  publisher={IEEE}
}

@article{vu2007linear,
  title={MIMO wireless linear precoding},
  author={Vu, Mai and Paulraj, Arogyaswami},
  journal={IEEE Signal Processing Magazine},
  volume={24},
  number={5},
  pages={86--105},
  year={2007},
  publisher={IEEE}
}

@article{Hoydis2013massivemimo,
  title={Massive MIMO in the UL/DL of cellular networks: How many antennas do we need?},
  author={Hoydis, Jakob and Ten Brink, Stephan and Debbah, M{\'e}rouane},
  journal={IEEE Journal on selected Areas in Communications},
  volume={31},
  number={2},
  pages={160--171},
  year={2013},
  publisher={IEEE}
}

@article{zhuang2025DBP,
  title={Decentralized MIMO systems with imperfect CSI using LMMSE receivers},
  author={Zhuang, Zeyan and Zhang, Xin and Xu, Dongfang and Song, Shenghui and Eldar, Yonina C},
  journal={IEEE Journal of Selected Topics in Signal Processing},
  year={2025},
  publisher={IEEE}
}

@article{zhuang2025no,
  title={No eigenvalues outside the limiting support of generally correlated and noncentral sample covariance matrices},
  author={Zhuang, Zeyan and Zhang, Xin and Xu, Dongfang and Song, Shenghui},
  journal={IEEE Transactions on Information Theory},
  year={2026},
  publisher={IEEE}
}

@ARTICLE{Junzhang2013,
  author={Zhang, Jun and Wen, Chao-Kai and Jin, Shi and Gao, Xiqi and Wong, Kai-Kit},
  journal={IEEE Journal on Selected Areas in Communications}, 
  title={On Capacity of Large-Scale MIMO Multiple Access Channels with Distributed Sets of Correlated Antennas}, 
  year={2013},
  volume={31},
  number={2},
  pages={133-148},
}

@ARTICLE{Wagner2012Seb,
  author={Wagner, Sebastian and Couillet, Romain and Debbah, Mérouane and Slock, Dirk T. M.},
  journal={IEEE Transactions on Information Theory}, 
  title={Large System Analysis of Linear Precoding in Correlated MISO Broadcast Channels Under Limited Feedback}, 
  year={2012},
  volume={58},
  number={7},
  pages={4509-4537},
}

@ARTICLE{Wen2007,
  author={Wen, Chao-Kai and Wong, Kai-Kit},
  journal=TIT, 
  title={Asymptotic Analysis of Spatially Correlated {MIMO} Multiple-Access Channels With Arbitrary Signaling Inputs for Joint and Separate Decoding}, 
  year={2007},
  volume={53},
  number={1},
  pages={252-268},
  month           = {Jan.},
}

@ARTICLE{Dumount2010Rician,
  author={Dumont, Julien and Hachem, Walid and Lasaulce, Samson and Loubaton, Philippe and Najim, Jamal},
  journal=TIT, 
  title={On the Capacity Achieving Covariance Matrix for {R}ician {MIMO} Channels: An Asymptotic Approach}, 
  year={2010},
  volume={56},
  number={3},
  pages={1048-1069},
  month           = {Mar.},
}

@ARTICLE{zhuang2025,
  author={Zhuang, Zeyan and Zhang, Xin and Xu, Dongfang and Song, Shenghui},
  journal={IEEE Transactions on Information Theory}, 
  title={Fundamental Limits of Two-Hop MIMO Channels: An Asymptotic Approach}, 
  year={2025},
  volume={71},
  number={6},
  pages={4069-4111},
}

@ARTICLE{Coui2011MIMOMAC,
  author={Couillet, Romain and Debbah, Mérouane and Silverstein, Jack W.},
  journal={IEEE Transactions on Information Theory}, 
  title={A Deterministic Equivalent for the Analysis of Correlated MIMO Multiple Access Channels}, 
  year={2011},
  volume={57},
  number={6},
  pages={3493-3514},
}

@ARTICLE{Xin2022,
  author={Zhang, Xin and Yu, Xianghao and Song, S.H.},
  journal={IEEE Journal of Selected Topics in Signal Processing}, 
  title={Outage Probability and Finite-SNR DMT Analysis for IRS-Aided MIMO Systems: How Large IRSs Need to be?}, 
  year={2022},
  volume={16},
  number={5},
  pages={1070-1085},
}
\endgroup

\end{document}